\documentclass[10pt,a4paper]{amsart}
\usepackage[utf8]{inputenc}
\usepackage{amsthm}
\usepackage{amssymb}
\usepackage{color}
\usepackage{mathabx}

\title[Structure of colored HOMFLY-PT for torus knots]{Combinatorial structure of colored HOMFLY-PT polynomials for torus knots}

\author[P.~Dunin-Barkowski]{P.~Dunin-Barkowski}

\address[P.~Dunin-Barkowski]{Faculty of Mathematics, National Research University Higher School of Economics, Usacheva 6, 119048 Moscow, Russia; and ITEP, Moscow 117218, Russia}
\email{ptdunin@hse.ru}

\author[A.~Popolitov]{A.~Popolitov}

\address[A.~Popolitov]{Department of Physics and Astronomy, Uppsala University, Uppsala, Sweden; Institute for Information Transmission Problems, Moscow 127994, Russia; and ITEP, Moscow 117218, Russia}
\email{popolit@gmail.com}

\author[S.~Shadrin]{S.~Shadrin}

\address[S. Shadrin]{Korteweg-de Vriesinstituut voor Wiskunde, 
	Universiteit van Amsterdam, Postbus 94248,
	1090GE Amsterdam, Nederland}
\email{s.shadrin@uva.nl}

\author[A.~Sleptsov]{A.~Sleptsov}

\address[A.~Sleptsov]{ITEP, Moscow 117218, Russia; Institute for Information Transmission Problems,	Moscow 127994, Russia; and Laboratory of Quantum Topology, Chelyabinsk State University, Chelyabinsk 454001, Russia}
\email{sleptsov@itep.ru}

\newcommand{\ba}[1]{\left \langle #1 \right \rangle}
\newcommand{\br}[1]{\left( #1 \right)}

\newcommand{\bs}[1]{\left [ #1 \right ]}

\newcommand{\txi}{\tilde{\xi}}

\newcommand{\al}{\alpha}
\newcommand{\be}{\beta}
\newcommand{\rhb}{\rho \beta}
\newcommand{\la}{\lambda}
\newcommand{\lb}{\left (}
\newcommand{\rb}{\right )}
\newcommand{\pb}{\bar{p}}
\newcommand{\pt}{\tilde{p}}
\newcommand{\ext}{\mathrm{ext}}
\newcommand{\Ze}{Z^{\ext}}
\newcommand{\Ec}{\mathcal{E}}
\newcommand{\Ac}{\mathcal{A}}
\newcommand{\At}{\widetilde{\Ac}}
\newcommand{\Fc}{\mathcal{F}}

\newcommand{\thephi}{\phantom{.}_2 \phi_1}
\newcommand{\theF}{\phantom{.}_2 F_1}
\newcommand{\cV}{\mathcal{V}}
\newcommand{\Z}{\mathbb{Z}}

\newcommand{\corc}[1]{\bigg{\langle} \, #1 \,  \bigg{\rangle} ^{\circ}}
\newcommand{\cord}[1]{\bigg{\langle} \, #1 \, \bigg{\rangle}}
\newcommand{\ac}{\mathfrak{A}}

\DeclareMathOperator{\Res}{Res}

\DeclareMathOperator{\Id}{Id}

\newtheorem{theorem}{Theorem}[section]
\newtheorem{proposition}[theorem]{Proposition}
\newtheorem{corollary}[theorem]{Corollary}
\newtheorem{lemma}[theorem]{Lemma}
\theoremstyle{definition}
\newtheorem{notation}[theorem]{Notation}
\newtheorem{remark}[theorem]{Remark}
\newtheorem{definition}[theorem]{Definition}

\newcommand{\p}[0]{\partial}
\newcommand{\binomial}[2]{
  \br{\begin{array}{c} #1 \\ #2 \end{array}}
}
\newcommand{\jp}[3]{\mathcal{P}_{#1}^{(#2,#3)}}

\begin{document}

\begin{abstract}
We rewrite the (extended) Ooguri-Vafa partition function for colored HOMFLY-PT polynomials for torus knots in terms of the free-fermion (semi-infinite wedge) formalism, making it very similar to the generating function for double Hurwitz numbers. This allows us to conjecture the combinatorial meaning of full expansion of the correlation differentials obtained via the topological recursion on the Brini-Eynard-Mari\~no spectral curve for the colored HOMFLY-PT polynomials of torus knots.

This correspondence suggests a structural combinatorial result for the extended Ooguri-Vafa partition function. Namely, its coefficients should have a quasi-polynomial behavior, where non-polynomial factors are given by the Jacobi polynomials. We prove this quasi-polynomiality in a purely combinatorial way. In addition to that, we show that the (0,1)- and (0,2)-functions on the corresponding spectral curve are in agreement with the extension of the colored HOMFLY-PT polynomials data.
\end{abstract}

\maketitle

\tableofcontents

\section{Introduction}

In 1989 Witten considered the 3D Chern-Simons quantum field theory and pointed out the existence of polynomial invariants of knots and links colored by the representations of Lie groups~\cite{Witten}. Based on his paper Reshetikhin and Turaev in~\cite{ReshetikhinTuraev} defined these polynomial invariants rigorously with the help of quantum groups and the R-matrix approach. The colored HOMFLY-PT polynomial corresponds to the $\mathfrak{sl}(N)$ case that plays a central role in the theory of knot polynomials.

In the case of torus knots, a formula for the colored HOMFLY-PT polynomial has a clear symmetric group character interpretation, called the Rosso-Jones formula. It was first derived in~\cite{RossoJones} and later, independently, in~\cite{LZ06}.
The Rosso-Jones formula allows to embed the colored HOMFLY-PT polynomials of torus knots in the realm of KP integrability, as it is done in~\cite{MMM12,MMS13}. In particular, it allows to represent the Ooguri-Vafa partition function \cite{OV} as an action of cut-and-join operators on a trivial KP tau-function \cite{MMSHurw13}. In the case of torus knots, this representation is especially simple.

Our first result relates explicitly the Ooguri-Vafa partition function of~\cite{MMS13} to a particular specialization of the partition function of double Hurwitz numbers~\cite{Okounkov}. This way we get a representation of the colored HOMFLY-PT polynomials in terms of the semi-infinite wedge formalism~\cite{Kac}, but in fact we get a more general formula that has more parameters. To be more precise, our extended partition function $Z^{ext}$ depends on the variables $\tilde p_m$, $m=1,2,\dots,$ and the original Ooguri-Vafa partition function  of the $(Q,P)$-torus knot is obtained by specialization $\tilde p_m=0$ for $m$ not divisible by $Q$:
\begin{align*}
Z^{\ext}(\pt) = 
\left\langle
\exp \lb 
\sum_{i=1}^\infty 
\frac{A^i-A^{-i}}
{
	e^{\frac{iuQ}{2P}}-
	e^{-\frac{iuQ}{2P}}
} 
\cdot
\frac{\alpha_{i}}{i} 
\rb
\exp \lb u \Fc_2 \rb
\exp \lb \sum_{j=1}^\infty \frac{\alpha_{-j} \pt_{j}}{j} \rb
\right\rangle.
\end{align*}
Here $A$ is the parameter of the HOMFLY-PT polynomial, and $u$ is the parameter that controls the genus expansion of $\log Z^{ext}$. The necessary definitions from the semi-infinite wedge formalism, in particular, the meaning of the operators $\alpha_{\pm j}$ and $\mathcal{F}_2$, we recall below.

The specialization of the partition function for double Hurwitz number that we obtain as an extension of the Ooguri-Vafa partition function is most natural to consider in the realm of the spectral curve topological recursion~\cite{EynardOrantin}, see also~\cite{AMM-1,AMM-2,CE,CEO,EO-2}.
Brini, Eynard, and Mari\~no represented in~\cite{BEM11} the coefficients of the genus expansion of the logarithm of the Ooguri-Vafa partition function for the $(Q,P)$-knot as particular coefficients of formal expansion of certain differential forms $\omega_{g,n}$ satisfying topological recursion on the curve 
\begin{align*} 
x(U) &= Q \log U + P \left[\log \left ( 1 - A^{P/Q+1} U\right ) - \log \left ( 1 - A^{P/Q-1} U\right )\right];\\
y(U) &= \gamma \log U + \dfrac{P\gamma + 1}{Q} \left[\log \left ( 1 - A^{P/Q+1} U\right ) - \log \left ( 1 - A^{P/Q-1} U\right ) \right].
\end{align*}
Here $U$ is a global coordinate on the rational curve, and $\gamma$ is an integer number chosen such that $P\gamma + 1$ is divisible by $Q$. One has to consider the expansions of $\omega_{g,n}$ in $\Lambda:=X^{-1/Q}$, $X=\exp(x)$, near the point $X=\infty$, and the coefficients of the Ooguri-Vafa partition function correspond to the integer powers of $X^{-1}$, namely, the coefficient of $u^{2g-2+n}\prod_{i=1}^n \tilde p_{Ql_i}$ in $\log Z^{ext}$ is equal, up to a universal combinatorial factor, to the coefficient of $\prod_{i=1}^n \Lambda_i^{Ql_i}$ in the expansion of $\omega_{g,n}$.

The status of this claim of Brini-Eynard-Mari\~no is the following. In~\cite{BEM11} it is formulated as a conjecture, and it is derived using a matrix model representation of the colored HOMFLY-PT polynomials of torus knots. Since that time the matrix models technique was developed in a number of papers, see~\cite{BorotEynardOrantin,BorotEynardWeisse,BorotGuionnetKozlowski}, and these works make the derivation of Brini-Eynard-Mari\~no mathematically rigorous.

We conjecture that our extension of the Ooguri-Vafa partition function satisfies the topological recursion on the same spectral curve. Namely, we conjecture that the coefficient of $u^{2g-2+n}\prod_{i=1}^n \tilde p_{m_i}$ in $\log Z^{ext}$ is equal, up to a universal combinatorial factor, to the coefficient of $\prod_{i=1}^n \Lambda_i^{m_i}$ in the expansion of $\omega_{g,n}$ for \emph{all} $m_1,\dots,m_n$, not only those divisible by $Q$. We verify this conjecture for the unstable terms $(g,n)=(0,1)$ and $(0,2)$. 

The main result of this paper is a quasi-polynomiality statement for the coefficients  $u^{2g-2+n}\prod_{i=1}^n \tilde p_{m_i}$ in $\log Z^{ext}$, for any $m_1,\dots,m_n$. This includes new combinatorial structural results for the colored HOMFLY-PT polynomials of torus knots, but in fact it is a more general statement.  We prove that these coefficients can be represented as 
\begin{align*}
& \sum_{
\substack{
	(i_1,\dots,i_n)
	\\
	\in \{0,1\}^n
}
} \mathcal{C}_{i_1,\dots,i_n}(m_1,\dots,m_n)
 \prod_{j=1}^n 
(-1)^{m_j} A^{m_j(\frac PQ-1)} \mathcal{P}_{m_j-1-i_j}^{(m_j(\frac PQ-1) + i_j, 1)}(1 - 2 A^2),
\end{align*}
where $\mathcal{P}_n^{\alpha,\beta}(x)$ are the Jacobi polynomials, and $\mathcal{C}_{i_1,\dots,i_n}$ are polynomials of finite degree that depends only on $g$ and $n$. 
Substituting $m_1,\dots,m_n$ divisible by $Q$ in this formula, we obtain a combinatorial structural result for the colored HOMFLY-PT of the $(Q,P)$-torus knot. 

There are several sources of interest for this type of expansions. First observation of this type in Hurwitz theory goes back to ELSV-formula~\cite{ELSV} for simple Hurwitz numbers given that represents a simple Hurwitz number in genus $g$ depending on $n$ multiplicities $m_1,\dots,m_n$ as 
$\mathcal{C}(m_1,\dots,m_n) \prod_{j=1}^n \frac{m_i^{m_i}}{m_i!}$, for certain polynomials $\mathcal{C}=\mathcal{C}_{g,n}$ representing the intersection theory of the moduli space of curves. There are some further results and conjectures relating Hurwitz theory to the intersection theory of the moduli space of curves, see, for instance,~\cite{JPTseng,Zvonkine,ShaSpiZvo,BKLPS}. The question of principle importance is whether a quasi-polynomiality property, natural from the point of view of algebraic geometry behind Hurwitz theory, can be derived in a combinatorial way, since it is a purely combinatorial property of purely combinatorial objects. For instance, for single Hurwitz numbers it was an open question for 15 years, since the ELSV formula was first introduced, until it was settled in two different ways in~\cite{DuninKazarian,KLS} (see also a discussion in~\cite{Lewanski}). Nowadays there are more of purely combinatorial results of this type, see, for instance,~\cite{KLS}, where the conjectures of~\cite{DoKarev,DoManescu} are resolved, as well as~\cite{DLPS,KLPS} for more examples of this type of statements. 

The quasi-polynomial structure has appeared to be also very important from the point of view of topological recursion~\cite{EynardOrantin}. It is proved in~\cite{Eynard}, see also~\cite{DOSS}, that the correlation differential obtained by topological recursion have the structure given by 
$\prod_{j=1}^n d_{x_j} \sum
\mathcal{C}_{i_1,\dots,i_n}(\frac{d}{dx_1},\dots,\frac{d}{dx_n})
\prod_{j=1}^n \xi_{i_j}(x_j),$
where the indices $i_1,\dots,i_n$ enumerate the critical points of function $x$ on the curve, $\mathcal{C}_{i_1,\dots,i_n}$ are certain polynomials, and $\xi_i$ are certain canonically defined functions on the spectral curve. This structure of correlation differentials is equivalent to the quasi-polynomiality of their expansion in a suitable coordinate (for example, in $x$ near $x=0$, $x^{-1}$ near $x=\infty$, or $e^x$ near $e^x=0$). So, in particular in our case the quasi-polynomiality structure means that the $n$-point functions corresponding to $\log Z^{ext}$ are expansion in $e^{-x/Q}$ near $e^{-x/Q}=0$ of $\mathcal{C}_{i_1,\dots,i_n}(\frac{d}{dx_1},\dots,\frac{d}{dx_n})
\prod_{j=1}^n \xi_{i_j}(x_j),$ where $\xi_1$ and $\xi_2$ are some functions canonically associated to the curve.

In this framework it is interesting to review the previously known non-polynomial parts of the quasi-polynomial formulas in Hurwitz theory. All of them are proved to be connected to $\xi$-function expansions.  Typically we have $r\geq 1$ critical points of $x$, and there is a natural basis of $r$ $\xi$-functions, called \emph{flat basis} due to its connection to flat coordinates in the theory of Frobenius manifold~\cite{DNOPS-1,DNOPS-2}. We have~\cite{KLS,KLPS}: 
\begin{align*}
& q\text{-orbifold } p\text{-spin Hurwitz numbers } (r=pq): & & \frac{m^{\lfloor \frac m{pq} \rfloor}}{\lfloor \frac m{pq} \rfloor!} \\
& r\text{-orbifold monotone Hurwitz numbers:} & & \binom{m+ \lfloor \frac mr \rfloor}{ m } \phantom{\Bigg|}\\
& r\text{-orbifold strictly monotone Hurwitz numbers:} & & \binom{m-1}{\lfloor \frac mr \rfloor} \phantom{\Bigg|}
\end{align*}
So, in all these cases there is a unique formula for the non-polynomial part, which reflects the fact that for each $m$ there is just one function $\xi$ (out of $r$ functions that form a flat basis) that can contribute to the corresponding term of the expansion of the correlation differentials. 
 In the case of the specialization of double Hurwitz numbers that we consider in this paper, $r=2$, and there are two possible non-polynomial terms for each $m$: 
\begin{align*}
& (-1)^{m} A^{m(\frac PQ-1)} \mathcal{P}_{m-1}^{(m(\frac PQ-1), 1)}(1 - 2 A^2) \\
\text{and} \qquad 
& (-1)^{m} A^{m(\frac PQ-1)} \mathcal{P}_{m-2}^{(m(\frac PQ-1) + 1, 1)}(1 - 2 A^2)
\end{align*}
These non-polynomial terms are mixed in the expansions of correlators in all possible ways, which, along with the fact that non-polynomiality is given by the values of Jacobi polynomials, makes this case really very special and challenging. From that point of view, the relation of $Z^{ext} $ to the Ooguri-Vafa partition function for the colored HOMFLY-PT polynomials of the $(Q,P)$-torus knots is just an extra nice property.



\subsection{Structure of the paper}

In Section \ref{sec:sinfw} we recall basic facts regarding the free fermions/semi-infinite wedge formalism.

In Section \ref{sec:rjff} we extend the Ooguri-Vafa partition function for the colored HOMFLY-PT polynomials for torus knots in a way that allows us to rewrite it in terms of the semi-infinite wedge formalism, which makes it quite similar to the generating function of double Hurwitz numbers.

In Section \ref{sec:bemcurve} we recall the spectral curve introduced in \cite{BEM11} and we derive the form of the so-called $\xi$-functions (in the general context these functions were defined in \cite{Eynard}, see also~\cite{DOSS,LPSZ}) on this curve. These $\xi$-functions turn out to be related to Jacobi polynomials.

In Section \ref{sec:jacobi} we recall the definition and some known properties of Jacobi polynomials. Then we derive some new properties and relations which are needed for the subsequent sections.

In Section \ref{sec:aops} we express the coefficients $K_{\mu_1,\dots,\mu_n}$ of the expansion of the extended Ooguri-Vafa partition function as semi-infinite wedge correlators of products of certain operators called $\mathcal{A}$-operators. Then we prove that matrix elements of these operators are polylinear combinations of $\xi$-functions that are rational in the arguments of the operators, i.e. in $\mu_i$.

In Section \ref{sec:Apoly} we prove a theorem that states that the connected correlators of  $\mathcal{A}$-operators are polynomial in $\mu_1,\dots,\mu_n$. This theorem, together with its reformulation found in the subsequent section, is the main result of the present paper.

In Section \ref{sec:combprooftr} we provide an alternative formulation of the main theorem of the present paper. In this alternative form, our main theorem states that in the stable case, i.e., for $2g-2+n>0$, the connected $(g,n)$-coefficients of the extended Ooguri-Vafa partition function can be organized into multidifferentials on the spectral curve which are finite polylinear combinations of differentials of derivatives of the $\xi$-functions. This is a key ingredient required for a combinatorial proof of topological recursion in the vein of \cite{DuninKazarian,DLPS,KLS,KLPS}. 

In Section \ref{sec:unstcorr} we match the unstable connected $(g,n)$-correlators of the extended Ooguri-Vafa partition function with the basic data of the spectral curve. This are the initial conditions required for the full topological recursion statement.

In Section \ref{sec:qsc} we derive a differential-difference operator that annihilates the wave function obtained from the extended Ooguri-Vafa partition function by principle specialization. The dequantization of this differential-difference operator gives the spectral curve of Brini-Eynard-Mari\~no.

\subsection{Acknowledgments}
We thank Ga\"etan Borot and Andrea Brini for useful discussions clarifying for us the current status of the result of~\cite{BEM11}. We thank Tom Koornwinder for his advice on $q$-hypergeometric functions and Jacobi polynomials. In particular, Tom observed that some of our propositions have a natural interpretation and an alternative proof in terms of the so-called Meixner polynomials, so Remarks~\ref{rem:MeixnerRem1} and~\ref{rem:Meixner-2} are due to him.

S.S. was supported by the Netherlands Organization for Scientific Research.
A.P. and A.S. were supported by the Russian Science Foundation (Grant No. 16-12-10344). P.D.-B. was supported by RFBR grant 16-31-60044-mol\_a\_dk and partially supported by RFBR grant 18-01-00461.

\section{Semi-infinite wedge preliminaries} \label{sec:sinfw}

In this section we remind the reader of some facts about semi-infinite wedge space. For more details see e.~g.~\cite{DuninKazarian}.

Let $V$ be an infinite dimensional vector space with a basis labeled by the half integers. Denote the basis vector labeled by $m/2$ by $\underline{m/2}$, so $V = \bigoplus_{i \in \Z + \frac{1}{2}} \underline{i}$.

\begin{definition}
	The semi-infinite wedge space $\cV$ is the span of all wedge products of the form
	\begin{equation}\label{wedgeProduct}
	\underline{i_1} \wedge \underline{i_2} \wedge \cdots
	\end{equation}
	for any decreasing sequence of half integers $(i_k)$ such that there is an integer $c$ (called the charge) with $i_k + k - \frac{1}{2} = c$ for $k$ sufficiently large. We denote the inner product associated with this basis by~$(\cdot,\cdot)$.
	
	The zero charge subspace $\cV_0\subset\cV$ of the semi-infinite wedge space is the space of all wedge products of the form~\eqref{wedgeProduct} such that 
	\begin{equation}\label{zeroCharge}
	i_k + k = \frac{1}{2}
	\end{equation}
	for $k$ sufficiently large.
\end{definition}

\begin{remark}
	An element of $\cV_0$ is of the form
	$$
	\underline{\lambda_1 - \frac{1}{2}} \wedge \underline{\lambda_2 - \frac{3}{2}} \wedge \cdots
	$$
	for some integer partition $\lambda$. This follows immediately from condition~\eqref{zeroCharge}. Thus, we canonically have a basis for $\cV_0$ labeled by all integer partitions. 
\end{remark}

\begin{notation}
	We denote by~$v_\lambda$ the vector labeled by a partition~$\lambda$. The vector labeled by the empty partition is called the vacuum vector and denoted by $|0\rangle = v_{\emptyset} = \underline{-\frac{1}{2}} \wedge \underline{-\frac{3}{2}} \wedge \cdots$. 
\end{notation}

\begin{definition}
	If $\mathcal{P}$ is an operator on $\cV_0$, then we define the \emph{vacuum expectation value} of~$\mathcal{P}$ by
	$\left\langle \mathcal{P} \right\rangle := \langle 0 |\mathcal{P}|0\rangle$,
	where $\langle 0 |$ is the dual of the vacuum vector with respect to the inner product~$(\cdot,\cdot)$, and called the covacuum vector. We will also refer to these vacuum expectation values as (disconnected) \emph{correlators}.
\end{definition}

Let us define some operators on the infinite wedge space.

\begin{definition} Let $k$ be any half integer. Then the operator $\psi_k\colon \cV\to\cV$ is defined by
	$\psi_k \colon (\underline{i_1} \wedge \underline{i_2} \wedge \cdots) \ \mapsto \ (\underline{k} \wedge \underline{i_1} \wedge \underline{i_2} \wedge \cdots)$. It increases the charge by $1$.
	
	The operator $\psi_k^*$ is defined to be the adjoint of the operator $\psi_k$ with respect to the inner product~$(\cdot,\cdot)$.
\end{definition}

\begin{definition}
	The normally ordered products of $\psi$-operators are defined in the following way
	\begin{equation}\label{eq:Eoperdef}
	E_{ij}:={:}\psi_i \psi_j^*{:}\ := \begin{cases}\psi_i \psi_j^*, & \text{ if } j > 0 \\
	-\psi_j^* \psi_i & \text{ if } j < 0\ .\end{cases} 
	\end{equation}
	This operator does not change the charge and can be restricted to $\cV_0$. Its action on the basis vectors $v_\lambda$ can be described as follows: ${:}\psi_i \psi_j^*{:}$ checks if $v_\lambda$ contains $\underline{j}$ as a wedge factor and if so replaces it by $\underline{i}$. Otherwise it yields~$0$. In the case $i=j > 0$, we have ${:}\psi_i \psi_j^*{:}(v_\lambda) = v_\lambda$ if $v_\lambda$ contains $\underline{j}$ and $0$ if it does not; in the case  $i=j < 0$, we have ${:}\psi_i \psi_j^*{:}(v_\lambda) = - v_\lambda$ if $v_\lambda$ does not contain $\underline{j}$ and $0$ if it does. These are the only two cases where the normal ordering is important.
\end{definition}

\begin{definition}
	Let $n \in \Z$ be any integer. We define an operator $\mathcal{E}_n(z)$ depending on a formal variable~$z$ by
	\begin{align}\label{eq:curlyEdef}
	\mathcal{E}_n(z) &:= \sum_{k \in \Z + \frac{1}{2}} e^{z(k - \frac{n}{2})} E_{k-n,k} + \frac{\delta_{n,0}}{\zeta(z)}  .
	\end{align}
	In some situations it is also useful to consider the operator $\widetilde{\mathcal{E}}_n(z)$ given by 
	\begin{align}\label{eq:curlyEtildedef}
\widetilde{\mathcal{E}}_n(z) &:=\sum_{k \in \Z + \frac{1}{2}} e^{z(k - \frac{n}{2})} E_{k-n,k} .
\end{align}
So, $\widetilde{\mathcal{E}}_n(z)=\mathcal{E}_n(z)$ for $n\not=0$, and in the case $n=0$ the difference of these two operators is the scalar multiplication by $\zeta(z)^{-1}$.  
\end{definition}

A useful formula for the commutator of these operators is given in~\cite{OkoPan}:
\begin{align}\label{eq:commEoper}
[\mathcal{E}_a(z),\mathcal{E}_b(w)]=\zeta(aw-bz)\mathcal{E}_{a+b}(z+w),
\end{align}
where $\zeta(z):=e^{z/2}-e^{-z/2}$. The same formula remains true if we use one or two $\widetilde{\mathcal{E}}$-operators on the left hand side instead of $\mathcal{E}$-operators.

\begin{definition}	\label{f2op} In what follows we will use the following operator:
\begin{align}		
{\mathcal F}_2 := \sum_{k\in\Z+\frac12} \frac{k^2}{2} E_{k,k} = [z^2]  \widetilde{\mathcal{E}}_0(z).
\end{align}
\end{definition}
Here $[z^2]$ denotes operation of taking the coefficient in front of $z^2$.
We use this notation throughout the paper.


\begin{definition}
	\label{alop}
	We will also need the following operators:
	$$\alpha_k:=\mathcal{E}_k(0), \quad k\neq 0 .$$
\end{definition}
Note the following fact regarding the commutators of these operators:
\begin{equation}\label{eq:alphacomm}
[\alpha_k,\alpha_m] = k\delta_{k+m,0}.
\end{equation}
Also note that
\begin{align}\label{eq:alphavac}
\alpha_k | 0 \rangle &= 0, \quad k>0; \\ \nonumber
\langle 0|\alpha_k &= 0, \quad k<0.
\end{align}

\section{Rosso-Jones formula as free-fermion average} \label{sec:rjff}


In this section we show that in the case of torus knots the Ooguri-Vafa generating function of the colored HOMFLY-PT polynomials can be extended, by adding additional parameters, to a partition
function that is natural from the point of view of the theory of double Hurwitz numbers.
The Ooguri-Vafa partition function is then recovered by setting all additional parameters to $0$.

The colored HOMFLY-PT polynomials of torus knots and links are known explicitly due to a result  of Rosso and Jones~\cite{RossoJones}. The HOMFLY-PT polynomials are knot invariants, i.e. they are invariant under the three Reidemeister moves defined on planar diagrams of knots. However, in many applications the HOMFLY-PT polynomials should be multiplied by an extra factor, which is not invariant under the first Reidemeister move. For this reason one can distinguish the \emph{topological} HOMFLY-PT polynomials and various \emph{non-topological} ones; different choices of common factors are called~\emph{framings}. Adapting the physics terminology, we can say that we consider HOMFLY-PT polynomials in topological framing, or in a so-called vertical one, or any other. In this paper we need a very specific choice of framing, which provides a consistency of knot invariants with the correlation differentials on the underlying spectral curve. We call it the \emph{spectral} framing.

\begin{definition}\label{torH}
	Let $T[Q,P]$ be a torus knot. Then its colored HOMFLY-PT polynomial in the spectral framing is defined by the following variation of the Rosso-Jones formula:
	\begin{equation}
	H_R(T[Q,P];q,A)	:= A^{P|R|} \cdot 
	\sum_{R_1 \vdash Q|R|} c^{R_1}_R\, q^{2\kappa_R \frac{P}{Q}}\, s^*_{R_1}(q,A),
	\end{equation}
	where the coefficients $c^{R_1}_R$ are integer numbers determined by the equation
	\begin{equation}\label{adams}
	\prod_{i=1}^l s_R(p_{Q|R|},p_{Q(|R|-1)},\dots,p_{Q})
	= \sum_{R_1 \vdash Q|R|}
	c^{R_1}_{R} \cdot s_{R_1}(p_{|R_1|},p_{|R_1|-1},\dots,p_{1}).
	\end{equation}
	(they are called Adams coefficients).
	The functions $s_R$ are the Schur polynomials, and $s^*_R := s_R|_{p_i=p^*_i}$, where
	\begin{align}
	\label{toplocus}
	p^*_i = \frac{A^i - A^{-i}}{q^{i} - q^{-i}}.
	\end{align} 
	(this substitution is called \emph{the restriction to the topological locus}).
\end{definition}

Now let us omit the restriction to the topological locus~\eqref{toplocus}  and define an extended colored HOMFLY-PT polynomial, which is no longer a knot invariant, but is still a braid invariant.
\begin{definition}
	\label{torH2}
	 The extended colored HOMFLY-PT polynomials in the spectral framing are defined as 
	\begin{equation}
	H_R(T[Q,P];p)	:= A^{P|R|} \cdot \sum_{R_1 \vdash Q|R|} c^{R_1}_R\, q^{2\kappa_R \frac{P}{Q}}\, s_{R_1}(p).
	\end{equation}
\end{definition}

\begin{definition}
	The Ooguri-Vafa partition function for the torus knot $T[Q,P]$ is
	\begin{align} 
	\label{eq:ov-partition-function}
	Z(p, \pb) := \sum_R  H_R(T[Q,P];p) \, s_R(\pb).
	\end{align}
\end{definition}
Note that we use the extended HOMFLY-PT polynomials in the definition of the partition function. The original Ooguri-Vafa partition function~\cite{OV} can be easily restored by the restriction to the topological locus~\eqref{toplocus}.

\begin{definition}
The topological	Ooguri-Vafa partition function for the torus knot  $T[Q,P]$  is defined by
\begin{equation}
Z(\pb) := Z(p^*,\pb)
\end{equation}
\end{definition}

Following~\cite{MMS13} we rewrite the Ooguri-Vafa partition function in terms of the cut-and-join operator. 
\begin{definition}
	The second cut-and-join operator $\hat{W}_2(p)$ in rescaled Miwa variables $p_k$, $k=1,2,\dots$, is defined as follows:
	\begin{equation}
	\label{W2}
	\hat{W}_2(p) = \frac{1}{2}\sum_{a,b=1}^\infty
	\left((a+b)p_ap_b\frac{\partial}{\partial p_{a+b}} +
	abp_{a+b}\frac{\partial^2}{\partial p_a\partial p_b}\right)
	\end{equation}
\end{definition}
Under the boson-fermion correspondence it corresponds to the operator $\mathcal{F}_2$, defined above. Its a very well-known property~\cite{Okounkov}, see also~\cite{MirMorNat}, that the Schur polynomials $s_R(p)$ are the eigenfunctions of this operator:
\begin{align}
\label{kappa}
\hat W_2(p) \, s_R(p) = \kappa_R \, s_R(p), \quad \kappa_R = \sum_{(i,j)\in R} (i-j).
\end{align}

\begin{lemma}[\cite{MMS13}] \label{lem:w2}
	The substitution $\hbar= 2\log q$ allows to rewrite the Ooguri-Vafa partition function for the torus knot $T[Q,P]$ in terms of cut-and-join operators:
	\begin{align} \label{eq:ovw2}
	Z(p, \pb)  = \exp \lb \hbar \frac{P}{Q} \hat{W}_2(p) \rb
	\exp \lb \sum_{k=1}^\infty \frac{p_{k Q} \pb_k A^{kP}}{k} \rb .
	\end{align}
\end{lemma}

\begin{proof}
	Let us substitute the formula for the HOMFLY-PT polynomials from Definition~\ref{torH2} in the Ooguri-Vafa partition function~\eqref{eq:ov-partition-function} and use the property~\eqref{kappa} of the cut-and-join operator. We have:
	\begin{align} 
	Z(p, \pb) & = \sum_R s_R(\pb)  A^{P|R|} \cdot \sum_{R_1 \vdash Q|R|} c^{R_1}_R\, q^{2\kappa_R \frac{P}{Q}}\, s_{R_1}(p)  
	\\ \notag
	& = \sum_R s_R(\pb)  A^{P|R|} \cdot \Big( q^{\hat{W}_2(p)\, \frac{P}{Q} } \cdot \sum_{R_1 \vdash Q|R|} c^{R_1}_R \, s_{R_1}(p) \Big) .
	\end{align}
%
	Then we use formula~\eqref{adams}:
	\begin{align} 
	Z(p,\pb)  & = \sum_R s_R(\pb)  \, A^{P|R|} \, q^{\hat{W}_2(p)\, \frac{P}{Q} } s_{R}(\{p_{kQ}\})   
	\\ \notag  
	& =q^{\hat{W}_2(p)\, \frac{P}{Q} } \sum_R s_R(\pb)  A^{P|R|} \,  s_{R}(\{p_{kQ}\})   
	\\ \notag
	& =  \exp \lb \hbar \frac{P}{Q} \hat{W}_2(p) \rb \exp \lb \sum_{k=1}^\infty \frac{p_{k Q} \pb_k A^{kP}}{k} \rb .
	\end{align}
\end{proof}

Our next step is to represent the Ooguri-Vafa partition function via an action of some operators on  the infinite wedge space. It is our starting point to reveal the combinatorial structure of colored HOMFLY-PT polynomials for torus knots.

\begin{theorem}\label{th:ov}
The Ooguri-Vafa partition function has the following representation  in terms of a semi-infinite wedge average with operators $\Fc_2$ and $\alpha_k$  introduced in the Definitions \ref{f2op} and \ref{alop}:
\begin{align} \label{eq:ov-rewritten}
Z(p, \pb)  = 
\left\langle
\exp \lb \sum_{j=1}^\infty \frac{\alpha_{-j} p_{j}}{j} \rb
\exp \lb \hbar \frac{P}{Q} \Fc_2 \rb
\exp \lb \sum_{i=1}^\infty \frac{\alpha_{iQ} (\pb_i \cdot QA^{iP})}{iQ} \rb
\right\rangle.
\end{align}
\end{theorem}

\begin{proof} Consider the expression given in Lemma~\ref{lem:w2}. The cut-and-join operator is applied to the function that can be represented as
\begin{align}
\exp \lb \sum_{k=1}^\infty \frac{p_{k  Q} \pb_kA^{kP}}{k} \rb =
\left\langle
\exp \lb \sum_{j=1}^\infty \frac{\alpha_{-j} p_{j Q} }{j} \rb
\exp \lb \sum_{i=1}^\infty \frac{\alpha_i \pb_iA^{iP}}{i} \rb
\right\rangle.
\end{align}
This representation follows from the commutation relations for the $\alpha$-operators~\eqref{eq:alphacomm} and their action on the vacuum and covacuum~\eqref{eq:alphavac}.

Observe, however, that we can represent it using a slightly different average
(using the $\alpha$-operators with the indices divisible by $Q$ and rescaling one of the denominators):
\begin{align}
\left\langle
\exp \lb \sum_{j=1}^\infty \frac{\alpha_{-j Q} p_{j Q}}{j Q} \rb
\exp \lb \sum_{i=1}^\infty \frac{\alpha_{i Q} \pb_iA^{iP}}{i} \rb
\right\rangle.
\end{align}
The sum in the left exponent now runs over $p$ with $Q$-divisible indices.
However, we can safely substitute it for the sum of all $p_j$'s -- thanks
to the commutation relations between $\alpha$'s the operators with the indices non divisible by $Q$ won't contribute. Thus, the OV partition function becomes
\begin{align} \label{eq:ov-rewritten1}
Z(p, \pb) & = \exp \lb \hbar \frac{P}{Q} \hat{W}_2(p) \rb
\left\langle
\exp \lb \sum_{j=1}^\infty \frac{\alpha_{-j} p_{j}}{j} \rb
\exp \lb \sum_{i=1}^\infty \frac{\alpha_{i Q} (\pb_i \cdot QA^{iP})}{iQ} \rb
\right\rangle \\ \notag
& =
\left\langle
\exp \lb \sum_{j=1}^\infty \frac{\alpha_{-j} p_{j}}{j} \rb
\exp \lb \hbar \frac{P}{Q} \Fc_2 \rb
\exp \lb \sum_{i=1}^\infty \frac{\alpha_{iQ} (\pb_i \cdot QA^{iP})}{iQ} \rb
\right\rangle,
\end{align}
where in the last equality we used the boson-fermion correspondence~\cite{JimboMiwa}
to bring the cut-and-join operator inside the average (where it becomes the operator $\Fc_2$).
\end{proof}
 
The expression~\eqref{eq:ov-rewritten} is already very close to a generating function
for double Hurwitz numbers -- except for a peculiar sum in the right exponential.
We can mend this by introducing \textit{additional parameters} $\pb_{i/Q}$, with
fractional indices.

\begin{definition}\label{def:extov}
The extended Ooguri-Vafa partition function is
 \begin{align} \label{eq:extended-partition-function}
 Z^{\ext}(p, \pb) &=
 \left\langle
 \exp \lb \sum_{j=1}^\infty \frac{\alpha_{-j} p_{j}}{j} \rb
 \exp \lb \hbar \frac{P}{Q} \Fc_2 \rb
 \exp \lb \sum_{i=1}^\infty \frac{\alpha_{i} (\pb_{i/Q} QA^{iP/Q})}{i} \rb
 \right\rangle
 \end{align}
\end{definition}
Note that $Z(p,\pb)$ is recovered from $Z^{\ext}(p, \pb)$ by keeping only $\pb_{i/Q}$ with integer indices.

Let us introduce the following useful piece of notation:
\begin{notation}
\begin{align}
\pt_i &:= \pb_{i/Q} QA^{iP/Q}, \quad i\geq 1, \\
u &:= \dfrac{P}{Q}\hbar,\\
\zeta(z)&:=e^{z/2}-e^{-z/2},\\
b &:= \dfrac{P}{Q},\\
a &:= A^2.
\end{align}
\end{notation}

Recall that we are interested in the restriction of $Z^{\ext}(p,\pb)$ to the topological locus, i.e. to $p=p^*$, where
\begin{align*}
p^*_i = \frac{A^i - A^{-i}}{e^{\hbar i/2} - e^{-\hbar i/2}} = \frac{A^i - A^{-i}}{\zeta(iuQ/P)}.
\end{align*}
Then, substituting new notations in Definition~\ref{def:extov} we obtain:
\begin{align}
Z^{\ext}(\pt) = 
\left\langle
\exp \lb \sum_{i=1}^\infty \frac{\alpha_{i} p^*_i}{i} \rb
\exp \lb u \Fc_2 \rb
\exp \lb \sum_{j=1}^\infty \frac{\alpha_{-j} \pt_{j}}{j} \rb
\right\rangle.
\end{align}
%
%

Let us expand $\Ze(\pt)$ in $\pb$. For an integer partition $\mu=(\mu_1\geq\dots\geq\mu_k)$  denote
\begin{equation} \label{eq:kmu-def}
K_\mu(u) := Q^{-n}\dfrac{\p}{\p\pb_{\mu_1/Q}}\dots \dfrac{\p}{\p\pb_{\mu_k/Q}}\Ze(\pb)\bigg|_{\pb=0},
\end{equation}
i.e. up to a factor equal to the order of the automorphism group of $\mu$
and a multiple of $Q$, $K_\mu(u)$ is the coefficient in front of the monomial $\pb_{\mu_1/Q}\dots \pb_{\mu_k/Q}$ in the expansion of $\Ze(\pb)$.


\section{The BEM spectral curve}\label{sec:bemcurve}
\subsection{Spectral curve}
Brini, Eynard and Mari\~no introduced in~\cite{BEM11} the following spectral curve associated to the torus knot $T[Q,P]$:
\begin{align} \label{eq:bem-curve}
  x(U) &= Q \log U + P \log \left ( 1 - A^{P/Q+1} U\right ) - P \log \left ( 1 - A^{P/Q-1} U\right )\\ \notag
  y(U) &= 
  \dfrac{\gamma}{Q} x(U)
  + \dfrac{1}{Q}\log \left ( 1 - A^{P/Q+1} U\right ) - \dfrac{1}{Q} \log \left ( 1 - A^{P/Q-1} U\right )
\end{align}
Here the constant $\gamma$ is chosen in such a way that
\begin{align}
\left(
\begin{array}{cc}
Q & P \\
\gamma & \delta
\end{array}
\right)
\in \mathrm{Sl}(2,\mathbb{Z}).
\end{align}
The functions $e^x$ and $e^y$ are meromorphic on the curve, and $U$ is a global coordinate (the curve has genus zero).
The Bergman kernel is the unique Bergman kernel for a genus zero curve and is equal to $$\frac{d U_1 dU_2}{(U_1 - U_2)^2}.$$
Denote
\begin{align} \label{eq:bem-x-and-lambda}
X&:= \exp(x)=U^Q \left (\dfrac{ 1 - A^{P/Q+1} U}{ 1 - A^{P/Q-1} U}\right )^P,
\\ \notag
\Lambda &:= X^{1/Q} = U \left (\dfrac{ 1 - A^{P/Q+1} U}{ 1 - A^{P/Q-1} U}\right )^{P/Q}.
\end{align}
The function $x(U)$ has two critical points $u_0 \pm \Delta u$ with
\begin{align}
u_0 = & \ \frac{1}{2 A^{b+1}} \br{1 + A^2 + b(A^2 - 1)} \\ \notag
(\Delta u)^2 = & \ \frac{1}{4 A^{2 b + 2}} \br{A^2 - 1}\br{(b+1)^2 A^2 - (b-1)^2}
\end{align}
%
It is convenient to use in the rest of the paper the variable 
\begin{align}
\lambda := \log \Lambda = \frac 1Q \log X.
\end{align}
%
%
The correlation multidifferentials $\omega_{g,n}(U_1, \dots, U_n)$ on the curve are defined through
the topological recursion procedure~\cite{EO-2}, see also~\cite{AMM-1,AMM-2,CEO}.
The claim of Brini, Eynard, and Mari\~no is that the expansion of $\omega_{g,n}$ at $X=\infty$ 
is related to $Z$ in the following way:
\begin{align}
\omega_{g,n} & \sim Q^{-n} \sum_{\mu_1 \dots \mu_n = 1}^\infty 
C^{(g)}_{Q\mu_1, \dots, Q\mu_n} \prod_{i = 1}^n (-\mu_i) \frac{d X_i}{X_i^{\mu_1 + 1}},
\end{align}
where similarity sign means that we drop all terms with non-integer exponents
of $X_i$ from expansion of $\omega_{g,n}$ at $X = \infty$, leaving
only terms with integer exponents of $X_i$.
The coefficients of these remaining terms are nothing but \textit{connected correlators}
for the partition function $Z$ \eqref{eq:ov-rewritten1}
\begin{align} 
C^{(g)}_{Q\mu_1, \dots, Q\mu_n} & := [\hbar^{2g-2+n}] \dfrac{\p}{\p\pb_{\mu_1}}\dots \dfrac{\p}{\p\pb_{\mu_n}}\log Z\bigg|_{\pb=0}.
\end{align}
Note that here we take derivatives with respect to $\pb$ and not $\pt$.

A natural conjectural refinement of their claim is the following extended relation of $Z^{ext}$ to the expansion of $\omega_{g,n}$ at $\Lambda=\infty$:
\begin{align}
\omega_{g,n} & = Q^{-n} \sum_{\mu_1 \dots \mu_n = 1}^\infty 
C^{(g)}_{\mu_1 \dots \mu_n} \prod_{i = 1}^n (-\mu_i) \frac{d \Lambda_i}{\Lambda_i^{\mu_i + 1}},
\end{align}
where $C^{(g)}_{\mu_1 \dots \mu_n}$ are now connected correlators
for the partition function $Z^{ext}$
\begin{align}  \label{eq:definitionCgmu}
C^{(g)}_{\mu_1 \dots \mu_n} & := [\hbar^{2g-2+n}] \dfrac{\p}{\p\pb_{\mu_1/Q}}\dots \dfrac{\p}{\p\pb_{\mu_n/Q}}\log Z^{ext}\bigg|_{\pb=0}
\end{align}
Note that this conjectural refinement of the result of Brini-Eynard-Mari\~no is very natural in the context of the results of the present paper.

The \textit{free-energies} $F_{g,n}$ are certain generating functions
for the connected correlators and are defined by their expansion
at $\Lambda=\infty$:
\begin{align}
F_{g,n} := \sum_{\mu_1 \dots \mu_n = 1}^\infty 
C^{(g)}_{\mu_1 \dots \mu_n} \prod_{i = 1}^n \frac{1}{\Lambda_i^{\mu_i}}.
\end{align}
They are related to the primitives of $\omega_{g,n}$ with appropriately chosen integration
constants.

\subsection{The $\xi$-functions through Jacobi polynomials} \label{sec:xiJac}

In this section we compute the so-called $\xi$-functions for the BEM spectral curve. The $\xi$-functions play
a critical role in the topological recursion procedure. Namely, the correlation (multi)differentials $\omega_{g,n}$, $2g-2+n>0$, are
expressed as finite polylinear combinations of the differentials of $\xi^a$, the index $a$ labels the critical points of $x$, and their derivatives (\cite{Eynard}, see also~\cite{DOSS,LPSZ}):
\begin{align}
\omega_{g,n} (U_1,\dots,U_n) = \sum_{\substack{a_1\dots a_n \\ d_1 \dots d_n}} c_{d_1\dots d_n}^{a_1 \dots a_n} \prod_{i=1}^n d\br{\frac{d}{d\lambda_i}}^{d_i}  \xi^{a_i}(U_i)
\end{align}
In the case of a genus $0$ spectral curve with a global coordinate $U$, the functions $\xi^a(U)$ are defined as some functions that form a convenient basis in the space
spanned by $1/(u_a - U)$, where $u_a$ are critical points of $\lambda(U)$ and $U$ is the global coordinate
on the spectral curve.

We show that we can choose these basis functions to have 
the following expansion at $\Lambda = \infty$:
\begin{align} \label{eq:xi-functions-expansion}
\xi^1(U) &= \sum_{m = 1}^\infty (-1)^m A^{(b-1)m} \mathcal{P}_{m-1}^{(m(b-1), 1)}(1 - 2 A^2) \Lambda^{-m} \\ \notag
\xi^2(U) &= \sum_{m = 1}^\infty (-1)^m A^{(b-1)m} \mathcal{P}_{m-2}^{(m(b-1) + 1, 1)}(1 - 2 A^2) \Lambda^{-m},
\end{align}
where $\mathcal{P}_n^{(\alpha,\beta)}(x)$ are the Jacobi polynomials (we recall their definition and needed
properties in Section \ref{sec:def-jacobi} below). Namely, we prove the following

\begin{theorem}\label{thm:xi-functions} We have:
	\begin{align} \label{eq:xi-functions-nonperturbative}
	\xi^1(U) = & \ \dfrac{u_0}{(\Delta u)^2 A^{b+1}} \txi_0(U) + \frac{1}{(\Delta u)^2 A^{b+1}} \txi_1(U) \\ \notag
	\xi^2(U) = & \ -\dfrac{1}{(\Delta u)^2 A^{2 b + 2}} \txi_0(U),
	\end{align}
	where
	\begin{align}
	\txi_0(U) = \frac{1}{1 - \dfrac{(U - u_0)^2}{(\Delta u)^2}}, \ \ \
	\txi_1(U) = \frac{U - u_0}{1 - \dfrac{(U - u_0)^2}{(\Delta u)^2}}
	\end{align}
\end{theorem}
In particular, it is clear that they indeed form a basis in
the space spanned by the functions $1/(u_0 \pm \Delta u - U)$.

\begin{notation} \label{not:xi-coeffs}
It is convenient to introduce the following piece of notation
for coefficients of expansion of $\xi^i(U)$ in $\Lambda$:
\begin{align}
	\xi^1_m &:= [\Lambda^{-m}] \xi^1(U) = (-1)^m A^{(b-1)m} \mathcal{P}_{m-1}^{(m(b-1), 1)}(1 - 2 A^2), \\ \notag
	\xi^2_m &:= [\Lambda^{-m}] \xi^2(U) = (-1)^m A^{(b-1)m} \mathcal{P}_{m-2}^{(m(b-1) + 1, 1)}(1 - 2 A^2).
\end{align}
\end{notation}

\begin{proof}
	Let us calculate the expansions of $\txi_k(U)$ at $\Lambda = \infty$.
	Introducing scaled and shifted global coordinate
	\begin{align}
	\tilde{U} = \frac{U - u_0}{\Delta u}
	\end{align}
	we see that
	\begin{align}
	\frac{d}{d\lambda} = \dfrac{U (1 - A^{b+1} U) (1 - A^{b-1} U)}{A^{2 b} (\tilde{U}^2-1)(\Delta u)^2} \frac{d}{dU}.
	\end{align}
	Therefore we can express $\txi_k$ as
	\begin{align}
	\txi_k = \frac{(-1) A^{2 b} \br{\Delta u}^2}{U (1 - A^{b+1} U) (1 - A^{b-1} U)} \frac{d}{d \lambda} \br{\frac{\br{U-u_0}^{k+1}}{k+1}},\quad k = 0,1.
	\end{align}
	The coefficients of the expansion of $\txi_k$ at $\Lambda = \infty$ can be computed by integration by parts and passing
	to the $U$-plane:
	\begin{align}
	& \bs{\Lambda^{-\mu}} \txi_k =   \frac{1}{2 \pi I} \oint_{\Lambda,\infty} \txi_k \Lambda^{\mu-1} d\Lambda
	\\ \notag
	& =  \frac{1}{2 \pi I} \oint_{\Lambda,\infty} \Lambda^{\mu-1} \frac{(-1) A^{2 b} \br{\Delta u}^2}{U (1 - A^{b+1} U) (1 - A^{b-1} U)} \Lambda \frac{d}{d \Lambda} \br{\frac{\br{U-u_0}^{k+1}}{k+1}} d\Lambda
	\\ \notag
	& =  \frac{A^{2 b} \br{\Delta u}^2}{2 \pi I} \oint_{U, \infty}
	\frac{\br{U-u_0}^{k+1}}{\br{k+1}}
	d\br{\frac{U^{\mu-1} \br{1 - A^{b+1} U}^{\mu b - 1}}
		{\br{1 - A^{b-1} U}^{\mu b + 1}}} 
	\end{align}
	Using Lemma \ref{lem:shifted-genfunc-expansion} below and differentiating the expansion termwise,
	we can easily calculate this residue to find
	\begin{align}
	\bs{\Lambda^{-\mu}} \txi_0 = & \ \br{\Delta u}^2 (-1) A^{2(b+1)} \bs{A^{(b-1)\mu} (-1)^\mu \mathcal{P}_{\mu-2}^{((b-1)\mu+1,1)} (1-2 A^2)};
	\\ \notag
	\bs{\Lambda^{-\mu}} \txi_1 = & \ \br{\Delta u}^2 A^{b+1} \bs{A^{(b-1)\mu} (-1)^\mu \mathcal{P}_{\mu-1}^{((b-1)\mu,1)} (1-2 A^2)}
	\\ \notag & \ + u_0 \br{\Delta u}^2 A^{2(b+1)} \bs{A^{(b-1)\mu} (-1)^\mu \mathcal{P}_{\mu-2}^{((b-1)\mu+1,1)} (1-2 A^2)},
	\end{align}
	which immediately implies the statement of the theorem.
\end{proof}

\begin{lemma} \label{lem:shifted-genfunc-expansion} We have:
	\begin{align}
	\frac{(1 - A_+ U)^{b \mu - y}}{(1 - A_- U)^{b \mu + y}} = \frac{A^{2 b (\mu - y)}}{U^{2 y}} \sum_{m=0}^\infty \frac{(-1)^m}{U^m A_+^m} \mathcal{P}_m^{(b \mu - y - m, 2 y - 1)}(1 - 2 A^2).
	\end{align}
\end{lemma}

\begin{proof} It is a direct computation:
	\begin{align}
	& \frac{(1 - A_+ U)^{b \mu - y}}{(1 - A_- U)^{b \mu + y}}
	 = \frac{(- A_+ U)^{b\mu-y}}{(-A_- U)^{b\mu+y}}
	\frac{\br{1 - \frac{1}{A_+ U}}^{b\mu - y}}{\br{1 - \frac{1}{A_- U}}^{b\mu + y}}
	\\ \notag
	& = \frac{A^{2 b (\mu - y)}}{U^{2 y}} \bs{\sum_{k=0}^\infty \br{-\frac{1}{A_+ U}}^k \binomial{b \mu - y}{k}}
	\bs{\sum_{l=0}^\infty \br{-\frac{1}{A_- U}}^l \binomial{- b \mu - y}{l}}
	\\ \notag
	& = \frac{A^{2 b (\mu - y)}}{U^{2 y}} \sum_{m=0}^\infty \frac{(-1)^m}{U^m A_+^m} \mathcal{P}_m^{(b \mu - y - m, 2 y - 1)}(1 - 2 A^2),
	\end{align}
	where in the last transition we use the definition of the Jacobi polynomials.
\end{proof}

\section{Jacobi polynomials and their generating functions} \label{sec:jacobi}
In the subsequent sections we heavily use certain results related to Jacobi polynomials which we obtain in the present section.
\subsection{Definition of Jacobi polynomials} \label{sec:def-jacobi}
Let us start with recalling the well-known  definition of the Jacobi polynomials.
\begin{definition}\label{def:jacobi} The Jacobi polynomials are defined as follows:
	\begin{align}
		\mathcal{P}_n^{(\al,\be)}(z):=\dfrac{\Gamma(\al+n+1)}{\Gamma(\al+1)\Gamma(n+1)}\theF\left(-n,1+\al+\be+n;\al+1;\dfrac{1}{2}(1-z)\right),
	\end{align}
	where
	\begin{equation}
	\theF(a,b;c;x):=\sum_{s=0}^{\infty}\dfrac{(a)_s\, (b)_s}{(c)_s\; s!}x^s
	\end{equation}
	is the usual hypergeometric function, where
	\begin{equation}
	(z)_n:=\dfrac{\Gamma[z+n]}{\Gamma[z]}=\left\{\begin{array}{ll}
	z(z+1)\cdots(z+n-1),& n\in\mathbb{Z}_{>0}\\
	1, & n=0
	\end{array}\right.
	\end{equation}
	is the usual Pochhammer symbol (rising factorial).    
\end{definition}

Let us introduce the following piece of notation:
\begin{notation}
	\begin{equation}\label{eq:Jdef}
	J_m(\rho):=\mathcal{P}_{m}^{(\rho b-m-1,1)}(1-2a),
	\end{equation}
	where $\mathcal{P}_n^{\al,\be}(x)$ is the Jacobi polynomial defined above.
\end{notation}

The functions $J_m(\rho)$ are precisely the main object of our study in this section.

\subsection{Three-term relation}

We found a certain new relation for the particular Jacobi polynomials $J_m(\rho)$ introduced in the previous subsection. We will need this relation in what follows. We call it the \emph{three-term relation}:
\begin{proposition}\label{prop:JacobiThreeTerm}
	For $k\in\mathbb{Z}_{\geq 1}$
	\begin{equation}\label{eq:JacobiThreeTerm}
	J_k+\left(a+1+(a-1)\dfrac{\rho b}{k}\right)J_{k-1}+aJ_{k-2}=0
	\end{equation}
\end{proposition}
\begin{proof}
	The following relations for Jacobi polynomials are well-known (see e.g. \cite{WTF}):
	\begin{align}
		\label{eq:JacobiPRel5}
		(n+1)\jp{n+1}{\al}{\be}(x) &= (n+\al+1)\jp{n}{\al}{\be}(x) - (n+\dfrac{\al}{2}+\dfrac{\be}{2}+1)(1-x)\jp{n}{\al+1}{\be}(x)\\ \label{eq:JacobiPRel8}
		(n+\be)\jp{n-1}{\al}{\be}(x) &= (n+\al+\be)\jp{n}{\al}{\be}(x) - (2n+\al+\be)\jp{n}{\al-1}{\be}(x)
	\end{align}
	In the remaining part of the proof we always put $x=1-2a$ in the arguments of the Jacobi polynomials and omit these arguments for brevity.
	
	Let us consider the expression in the LHS of \eqref{eq:JacobiThreeTerm}, and let us transform it with the help of \eqref{eq:JacobiPRel5} and \eqref{eq:JacobiPRel8}:
	\begin{align*}
		&\jp{k}{\rhb -k -1}{1}+\left(a+1+(a-1)\dfrac{\rho b}{k}\right)\jp{k-1}{\rhb -k}{1}+a\jp{k-2}{\rhb -k+1}{1}\\
		&\mathop{=}_{\eqref{eq:JacobiPRel5}} \dfrac{\rhb-1}{k}\jp{k-1}{\rhb -k-1}{1}+\dfrac{k-\rhb}{k}\jp{k-1}{\rhb -k}{1}+a\jp{k-2}{\rhb -k+1}{1}\\
		&\mathop{=}_{\eqref{eq:JacobiPRel5}} \dfrac{\rhb-1}{k}\jp{k-1}{\rhb -k-1}{1}+\dfrac{\rhb(1-\rhb)}{k(\rhb+k-1)}\jp{k-1}{\rhb -k}{1}+\dfrac{\rhb-1}{\rhb+k-1}\jp{k-2}{\rhb -k}{1}\\
		&\mathop{=}_{\eqref{eq:JacobiPRel8}} \left(\dfrac{\rhb-1}{k}+\dfrac{\rhb-1}{\rhb+k-1}\dfrac{1-\rhb-k}{k}\right)\jp{k-1}{\rhb -k-1}{1}\\ &\phantom{\mathop{=}_{\eqref{eq:JacobiPRel8}} }+ \left(\dfrac{\rhb(1-\rhb)}{k(\rhb+k-1)}+\dfrac{\rhb-1}{\rhb+k-1}\dfrac{\rhb}{k}\right)\jp{k-1}{\rhb -k}{1}\\
		&=0
	\end{align*}
\end{proof}

\begin{remark}[Meixner polynomials I] \label{rem:MeixnerRem1}
The specialization of the Jacobi polynomials that we are using in this paper has a natural interpretation in terms of the so-called Meixner polynomials~\cite[Section 9.10]{Koekoek}. The Meixner polynomials are defined as
\[
M_{n}(x,\beta,c):=\theF(-n,x;\beta;1-c^{-1}).
\] 
So, we see that the polynomials $J_m(\rho)$ that play the main role in our analysis can be expressed in terms of the Meixner polynomials as 
\begin{align*}
J_m(\rho) & = \mathcal{P}_{m}^{(\rho b-m-1,1)}(1-2a) = (-1)^m \mathcal{P}_{m}^{(1,\rho b-m-1)}(2a-1) 
\\
& =(-1)^m (m+1) \theF(-m,1+\rho b;2;1-a)
\\
& = (-1)^m (m+1) a^m \theF(-m,1-\rho b;2;1-a^{-1})
\\
& = (-1)^m a^m (m+1) M_m(\rho b-1, 2, a).
\end{align*}
The recurrence relation for the Meixner polynomials~\cite[Equation (9.10.3)]{Koekoek} implies that 
\begin{align*}
& a(k+1)M_k(\rho b-1, 2, a)
\\
& -\left((a-1)(\rho b-1)+k-1+a(k+1)\right) M_{k-1}(\rho b-1, 2, a)
\\
& +(k-1)M_{k-2}(\rho b -1, 2, a) =0,
\end{align*}
which is equivalent to~\eqref{eq:JacobiThreeTerm}. This way we obtain an alternative proof of Proposition~\ref{prop:JacobiThreeTerm}.
\end{remark}

\subsection{Exponential generating functions for Jacobi polynomials}

In what follows we will need the following representation for the special Jacobi polynomials defined in~\eqref{eq:Jdef}:
\begin{proposition}\label{prop:j12exprs}
	For $m\in \mathbb{Z}_{\geq 1}$
	\begin{align}\label{eq:jm1expr}
		(-1)^m(1-a) \dfrac{\rho b}{m}J_{m-1}(\rho) &= [w^m]\exp\left(\sum_{i=1}^\infty \frac{a^i-1}{i}\cdot w^i \rho b\right)\\ \label{eq:jm2expr}
		(-1)^m(1-a) \dfrac{\rho b}{m}J_{m-2}(\rho) &= -\dfrac{m-1}{m}[w^{m-1}]\exp\left(\sum_{i=1}^\infty \frac{a^i-1}{i}\cdot w^i \rho b\right)
	\end{align}
\end{proposition}
\begin{proof}	
	From the definition of Jacobi polynomials (Definition \ref{def:jacobi}) it is easy to see that for $m\in \mathbb{Z}_{\geq 1}$ we have
	\begin{align}\label{eq:FormulaForJacobi}
		&(-1)^m (1-a) \dfrac {\rho b}{m} J_{m-1}(\rho)
		\\ \notag
		&= (-1)^m (1-a) \dfrac{\rho \be}{m} \cdot \sum_{s=0}^{m-1}
		\binom{\rho b-1}{s}\binom{m}{s+1} (-a)^{m-1-s} (1-a)^{s}
		\\ \notag
		&=  \sum_{s=0}^{m-1}
		\binom{\rho b}{s+1}\binom{m-1}{s} a^{m-1-s} (a-1)^{s+1}
	\end{align}
	Now let us collect this into the generating series:
	\begin{align} \label{eq:jlhscomp}
		&1+\sum_{m=1}^\infty(-1)^m(1-a) \dfrac{\rho b}{m}J_{m-1}(\rho) w^m 
		\\ \notag &
		= 1+ \sum_{m=1}^\infty \sum_{s=0}^{m-1}
		\binom{\rho b}{s+1}\binom{m-1}{s} a^{m-1-s} (a-1)^{s+1}w^m\\ \nonumber
		&= 1+ \sum_{s=0}^{\infty}
		\binom{\rho b}{s+1}(w(a-1))^{s+1} \sum_{k=0}^\infty \binom{s+k}{s} (wa)^{k} \\ \notag
		& = 1+ \sum_{s=0}^{\infty}
		\binom{\rho b}{s+1}\left(\frac{w(a-1)}{1-wa}\right)^{s+1} \\ \notag
		& = \left(1+ \frac{w(a-1)}{1-wa}\right)^{\rho b} 
		\\ \notag &
		 = \left(\frac{1-w}{1-wa}\right)^{\rho b},
	\end{align}
	where in the second equality we used that $m-1-s=k\geq 0$.
	
	The generating series for the RHS of \eqref{eq:jm1expr} has the following form:
	\begin{align}\label{eq:JacobiGSleft}
		&\exp \left(\sum_{i=1}^\infty \frac{a^i-1}{i}\cdot \rho b w^i\right)  = \exp \left(\rho b \sum_{i=1}^\infty \frac{(wa)^i}{i} - \rho b \sum_{i=1}^\infty \frac{w^i}{i}\right) \\ \nonumber
		&= \exp \left(-\rho b\log(1-wa)+\rho b \log(1-w)\right) 
		= \left(\frac{1-w}{1-wa}\right)^{\rho b},
	\end{align}
	which precisely coincides with the result of the computation for the LHS in \eqref{eq:jlhscomp}. This proves equality \eqref{eq:jm1expr}. Equality \eqref{eq:jm2expr} is an easy consequence of \eqref{eq:jm1expr}.
\end{proof}

\begin{remark}[Meixner polynomials II] \label{rem:Meixner-2} This proposition can also be interpreted in terms of the Meixner polynomials, as in Remark~\ref{rem:MeixnerRem1}. Indeed, Remark~\ref{rem:MeixnerRem1} allows to rewrite the generating function on the left hand side of Equation~\eqref{eq:jlhscomp} as 
\begin{equation}\label{eq:genfuncMeixner}
1+(a-1)\rho b w \sum_{m=0}^\infty  M_{m}(\rho b-1,2,a)  (aw)^m 
\end{equation}
There is a generating function for the Meixner polynomials, see~\cite[Equation (9.10.13) for $\gamma=1$]{Koekoek}. Its specialization to our parameters allows to prove that~\eqref{eq:genfuncMeixner} is equal to the right hand side of Equation~\eqref{eq:jlhscomp}, which leads to an alternative proof of Proposition~\ref{prop:j12exprs}.
\end{remark}

Proposition~\ref{prop:j12exprs} immediately implies the following statement:
\begin{corollary}
	\label{prop:Jacobiseries}
	For $m\in \mathbb{Z}_{\geq 1}$ we have:
	\begin{equation}
	\sum_{\lambda\vdash m} \prod_{i=1}^\infty \frac{\left(\frac{A^i-A^{-i}}{i}\cdot x \right)^{\lambda_i-\lambda_{i+1}}}{(\lambda_i-\lambda_{i+1})!} 
	= (-1)^m A^{-m} (1-A^2) \frac xm \cdot \jp{m-1}{x-m}{1}(1-2A^2)
	\end{equation}
\end{corollary}

We will also need the following important result:
\begin{proposition} \label{prop:rationality}
	For $m\in\mathbb{Z}_{\geq 1}$
	\begin{align}\label{eq:expJacobi}
	& [u^{2k}w^m]\exp\left(\sum_{i=1}^\infty \frac{a^i-1}{i}\cdot w^i \dfrac{\zeta(iu\rho)}{\zeta(iu b^{-1})}\right)
	\\ \notag &
	=
	(-1)^m \dfrac{\rho}{m}\left(G_k^1(\rho,m)J_{m-1}+G_k^2(\rho,m)J_{m-2}\right),
	\end{align}	
	where $G^1_k$ and $G^2_k$ are polynomials in $m$ and $\rho$ of degree no greater than $9k+2$. Terms with odd powers of $u$ in the $u$-expansion of the exponential in the LHS of the above equality vanish.
\end{proposition}


In order to prove this proposition first we need to go through certain technical reasoning. 

Denote:
\begin{notation}
	\begin{equation}
	q:=\exp\left(-\dfrac{u}{b}\right)
	\end{equation}
\end{notation} 
Recall the well-known definition of the $q$-Pochhammer symbols:
\begin{definition}
	\emph{$q$-Pochhammer symbols} are defined as follows:
	\begin{align}
		(x;q)_m := \prod_{k = 1}^{m} (1 - x q^{k-1})
	\end{align}
\end{definition}

Now let us prove the following
\begin{proposition} \label{prop:hypergeom}
	For $m\in\mathbb{Z}_{\geq 1}$
	\begin{align}\label{eq:hypergeomexpr}
		&[w^m]\exp\left(\sum_{i=1}^\infty \frac{a^i-1}{i}\cdot w^i \dfrac{\zeta(iu\rho)}{\zeta(iu b^{-1})}\right)=\\ \nonumber
		&=\left( q^{1/2 + \rho  b/2} \right)^m
		\frac{(q^{-\rho  b}; q)_m}{(q;q)_m} \thephi \left(q^{-m}, q^{\rho  b}; q^{\rho  b + 1 - m} ; q ; a q \right),
	\end{align}
	where $\thephi$ is the $q$-hypergeometric function
	\begin{align} \label{eq:q-hypergeom-def}
		\thephi \left( a_1, a_2; b_1; q; x \right)
		:= \sum_{n = 0}^\infty \frac{(a_1; q)_n (a_2; q)_n}{(b_1; q)_n (q; q)_n} x^n.
	\end{align}
\end{proposition}
\begin{proof}
	Denote
	\begin{equation}
	\Xi(w):=\exp\left(\sum_{i=1}^\infty \frac{a^i-1}{i}\cdot w^i \dfrac{\zeta(iu\rho)}{\zeta(iu b^{-1})}\right).
	\end{equation}
	Note that
	\begin{equation}
	\dfrac{\zeta(iu\rho)}{\zeta(iu b^{-1})} = \dfrac{q^{-i\rho b/2}-q^{i\rho b/2}}{q^{-i/2}-q^{i/2}}
	=q^{i(1-\rho b)/2}\;\dfrac{q^{i\rho b}-1}{q^{i}-1}.
	\end{equation}
	Denote
	\begin{equation}
	\sigma:=\rho b.
	\end{equation}
	Consider the case of positive integer $\sigma$. In this case
	\begin{equation}
	\dfrac{\zeta(iu\rho)}{\zeta(iu b^{-1})} 
	=q^{i(1-\sigma)/2}\;\dfrac{q^{i\sigma}-1}{q^{i}-1}=\sum_{p=1}^{\sigma}q^{i(2p-\sigma-1)/2}.
	\end{equation}
	Now (still in the case of $\sigma\in \mathbb{Z}_{>0}$)
	\begin{align}
		\Xi(w)&=\exp\left(\sum_{i=1}^\infty \frac{a^i-1}{i}\cdot w^i \dfrac{\zeta(iu\rho)}{\zeta(iu b^{-1})}\right)\\ \nonumber
		&=\exp\left(\sum_{p=1}^{\sigma}\left(\sum_{i=1}^\infty \dfrac{q^{i(2p-\sigma-1)/2} a^i w^i}{i}-\sum_{i=1}^\infty \dfrac{q^{i(2p-\sigma-1)/2} w^i}{i}\right)\right)\\ \nonumber
		&=\exp\left(\sum_{p=1}^{\sigma}\left(\log\left(1-wq^{(2p-\sigma-1)/2}\right)-\log\left(1-w a q^{(2p-\sigma-1)/2}\right)\right)\right)\\ \nonumber
		&=\dfrac{(wq^{(1-\sigma)/2};q)_\sigma}{(waq^{(1-\sigma)/2};q)_\sigma}\\ \nonumber
		&=\dfrac{(wq^{(1-\sigma)/2};q)_\infty(waq^{(1+\sigma)/2};q)_\infty}{(wq^{(1+\sigma)/2};q)_\infty(waq^{(1-\sigma)/2};q)_\infty}.
	\end{align}
	Recall the $q$-binomial formula:
	\begin{align}
		\frac{(\alpha x; q)_\infty}{(x; q)_\infty} = \sum_{n = 0}^\infty \dfrac{(\alpha; q)_n}{(q; q)_n} x^n.
	\end{align}
	With the help of this formula we get:
	\begin{align}
		\Xi(w)&=\dfrac{(wq^{(1-\sigma)/2};q)_\infty(waq^{(1+\sigma)/2};q)_\infty}{(wq^{(1+\sigma)/2};q)_\infty(waq^{(1-\sigma)/2};q)_\infty}\\ \nonumber
		&= \sum_{s = 0}^\infty \dfrac{(q^{-\sigma}; q)_s}{(q; q)_s} w^s q^{s(1+\sigma)/2}\;  \sum_{r = 0}^\infty \dfrac{(q^{\sigma}; q)_r}{(q; q)_r} a^r w^r q^{r(1-\sigma)/2}.
	\end{align}
	Collecting the powers of $w$ we get ($m=s+r$):
	\begin{equation}
	\Xi(w) = \sum_{m=0}^{\infty}\sum_{r=0}^{m} \dfrac{(q^{\sigma}; q)_r}{(q; q)_r}  \dfrac{(q^{-\sigma}; q)_{m-r}}{(q; q)_{m-r}} q^{-r\sigma}a^r q^{m(1+\sigma)/2} w^m.
	\end{equation}
	
	Now note the following fact:
	\begin{align}
		\frac{(x;q)_{m - r}}{(y;q)_{m - r}} = \frac{(x;q)_{m}}{(y;q)_{m}}
		\frac{(y^{-1} q^{1 - m};q)_{r}}{(x^{-1} q^{1 - m} ; q)_{r}} \left(\frac{y}{x}\right)^r.
	\end{align}
	This allows us to rewrite $\Xi(w)$ as follows:
	\begin{equation}\label{eq:hyperexser}
	\Xi(w) = \sum_{m=0}^{\infty} w^m  q^{m(1+\sigma)/2} \dfrac{(q^{-\sigma}; q)_{m}}{(q; q)_{m}} \sum_{r=0}^{m} \dfrac{(q^{\sigma}; q)_r(q^{-m}; q)_{r}}{(q; q)_r(q^{\sigma+1-m}; q)_{r}} q^{r}a^r.
	\end{equation}
	Note that taking the coefficient in front of $w^m$ and recalling the definition of the q-hypergeometric function \eqref{eq:q-hypergeom-def} allows us to immediately obtain the proof of \eqref{eq:hypergeomexpr} for the case of positive integer $\sigma$ (since the above arguments used this assumption).
	
	Now let us relax the assumption on $\sigma$ and denote
	\begin{equation}
	z:=q^{\sigma}.
	\end{equation}
	Recall from \eqref{eq:JacobiGSleft} that
	\begin{align}
		\Xi(w)&= \sum_{m=0}^\infty w^m \sum_{\lambda\vdash m} \prod_{i=1}^\infty \dfrac{1}{(\lambda_i-\lambda_{i+1})!}{\left(\frac{a^i-1}{i}\cdot q^{i(1-\sigma)/2}\;\dfrac{q^{i\sigma}-1}{q^{i}-1} \right)^{\lambda_i-\lambda_{i+1}}}\\ \nonumber
		&=\sum_{m=0}^\infty w^m q^{-m\sigma/2} \sum_{\lambda\vdash m} \prod_{i=1}^\infty \dfrac{1}{(\lambda_i-\lambda_{i+1})!}{\left(\frac{a^i-1}{i}\cdot q^{i/2}\;\dfrac{q^{i\sigma}-1}{q^{i}-1} \right)^{\lambda_i-\lambda_{i+1}}}.
		\end{align}
	Thus
	\begin{align}
		[w^m]\Xi(w)&=z^{-m/2} \sum_{\lambda\vdash m} \prod_{i=1}^\infty \dfrac{1}{(\lambda_i-\lambda_{i+1})!}{\left(\frac{a^i-1}{i}\cdot q^{i/2}\;\dfrac{z^{i}-1}{q^{i}-1} \right)^{\lambda_i-\lambda_{i+1}}}\\ \nonumber
		&=z^{-m/2} P_m(z),
	\end{align}
	where $P_m(z)$ is some polynomial in $z$ of degree $\leq m$.
	
	Now consider the right hand side of \eqref{eq:hypergeomexpr}. Let us denote it by $\Theta_m$. It is equal to
	\begin{equation}
	\Theta_m=q^{m(1+\sigma)/2} \dfrac{(q^{-\sigma}; q)_{m}}{(q; q)_{m}} \sum_{r=0}^{\infty} \dfrac{(q^{\sigma}; q)_r(q^{-m}; q)_{r}}{(q; q)_r(q^{\sigma+1-m}; q)_{r}} q^{r}a^r.
	\end{equation}
	Note that the sum actually runs until $m$ since for $r>m$ the factor $(q^{-m}; q)_{r}$ vanishes. We have
	\begin{align}
		\Theta_m&=\sum_{r=0}^{m} \dfrac{(q^{\sigma}; q)_r}{(q; q)_r}  \dfrac{(q^{-\sigma}; q)_{m-r}}{(q; q)_{m-r}} q^{-r\sigma}a^r q^{m(1+\sigma)/2} 
		\\ \notag
		&=\sum_{r=0}^{m} \dfrac{\prod_{k=1}^r (1-zq^{k-1})}{(q; q)_r}  \dfrac{\prod_{l=1}^{m-r} (z-q^{l-1})}{(q; q)_{m-r}} z^{r-m} q^{m/2}a^r z^{-r+m/2} \\ \nonumber
		&=z^{-m/2} Q_m(z),
	\end{align}
	where $Q_m(z)$ is some polynomial in $z$ of degree $\leq m$.
	
	From the above we know that $P_m(z)$ and $Q_m(z)$ coincide for all values of $z=q^\sigma$ where $\sigma$ is a positive integer. This means that these two polynomials of degrees $\leq m$ coincide at an infinite number of distinct points (all these points are distinct since $q=\exp(u/b)$, and $u$ is an arbitrarily small formal parameter, while $b$ is a fixed rational number). Thus, these polynomials coincide, which proves the proposition. 
\end{proof}



Suppose we have a function $f(q, m, \rho)$.
We can consider either $f(e^\hbar, m, \rho)$ and Taylor expansion in powers of $\hbar$,
or $f(1 + \epsilon, m, \rho)$ and Taylor expansion in powers of $\epsilon$ (in our case $\hbar = u b^{-1}$).
These expansions either both have the polynomiality property, or they both do not, in the following sense:

\begin{lemma}\label{lem:coefficients-h-e}
	All coefficients of $\hbar$-expansion of $f$ can be represented as a sum of two Jacobi polynomials
	with polynomial coefficients $Poly_{1,k}$ and $Poly_{2,k}$ in $m$ and $\rho$
	\begin{align}
		[\hbar^k] f(e^\hbar, m, \rho) = Poly_{1,k}(m, \rho) J_m + Poly_{2,k}(m, \rho) J_{m-1}
	\end{align}
	if and only if all coefficients of $\epsilon$-expansion of $f$ can be represented as a sum of two Jacobi polynomials
	with some other polynomial coefficients $\widetilde{Poly}_{1,k}$ and $\widetilde{Poly}_{2,k}$ in $m$ and $\rho$
	\begin{align}
		[\epsilon^k] f(1 + \epsilon, m, \rho) = \widetilde{Poly}_{1,k}(m, \rho) J_m + \widetilde{Poly}_{2,k}(m, \rho) J_{m-1},
	\end{align}
	and moreover in that case the degree of $Poly_{i,k}(m, \rho)$ is no greater than the degree of $\widetilde{Poly}_{i,k}$ and vice versa.
\end{lemma}

\begin{proof}
	Indeed, every coefficient of $\hbar$-expansion is a finite linear combination of the coefficients of $\epsilon$-expansion,
	where number of summands depends neither on $m$ nor $\rho$, from which the statement follows.
\end{proof}

Now let us look at the RHS of \eqref{eq:hypergeomexpr} and analyze the three factors in it separately. Let us start with the first factor, $\left( q^{1/2 + \rho b/2} \right)^m$.
\begin{lemma}
	For $q=1+\epsilon$, the coefficient in front of $\epsilon^k$ of the $\epsilon$-expansion of  $\left( q^{1/2 + \rho b/2} \right)^m$
	is a polynomial in $m$ and $\rho$ of total degree no greater than $2k$.
\end{lemma}

\begin{proof}
	Indeed
	\begin{align}
		\left( (1 + \epsilon)^{1/2 + \rho b/2} \right)^m
		= \sum_{k = 0}^\infty \epsilon^k \binomial{\frac{m + \rho b m}{2}}{k},
	\end{align}
	where each term is clearly polynomial, and the degree is evident.
\end{proof}

Let us proceed to the second factor in  formula \eqref{eq:hypergeomexpr}.
We have the following:
\begin{lemma} \label{lem:q-poch-poly}
	For $m\in\mathbb{Z}_{\geq 1}$, $q=1+\epsilon$,
	\begin{align}
		\frac{(q^{-\rho b}; q)_m}{(q;q)_m} = \frac{\Gamma(m - \rho b)}{\Gamma(m+1) \Gamma(-\rho b)}
		\left (1 + \sum_{k = 1}^\infty \epsilon^k Poly_{k}(m, \rho) \right),
	\end{align}
	i.e. modulo common non-polynomial factor coefficients of $\epsilon$-expansion
	of the second coefficient are polynomials in $m$ and $\rho$. Moreover, the total degree of $Poly_{k}(m, \rho)$ is no greater than $2k+1$.
\end{lemma}

\begin{proof}
	Indeed, it is easy to extract common non-polynomial factor
	\begin{align}
		& \frac{\prod_{i = 1}^m (1 - (1 + \epsilon)^{-\rho b + i - 1})}
		{\prod_{i = 1}^m (1 - (1 + \epsilon)^i)}
		 = \frac{\prod_{i = 1}^m \left( \sum_{k=1}^\infty \binomial{-\rho b + i - 1}{k} \epsilon^k \right)}
		{\prod_{i = 1}^m \left( \sum_{k=1}^\infty \binomial{i}{k} \epsilon^k \right)} \\ \notag
		& = \frac{\prod_{i=1}^m (-\rho b+i-1)}{\prod_{i=1}^m i}
		\frac{\prod_{i = 1}^m \left( \sum_{k=0}^\infty \epsilon^k \binomial{-\rho b + i - 1}{k+1}\frac{1}{(-\rho b+i-1)}  \right)}
		{\prod_{i = 1}^m \left( \sum_{k=0}^\infty \epsilon^k \binomial{i}{k+1}\frac{1}{i} \right)}.
	\end{align}
	Let's illustrate what happens by expanding the numerator up to second order in $\epsilon$:
	\begin{align}
		1 + \epsilon \sum_{i=1}^m \frac{-\rho b + i - 2}{2}
		+ \epsilon^2 \Bigg(&\sum_{i=1}^m \frac{(-\rho b + i - 2)(-\rho b + i - 3)}{3 \cdot 2}	
		\\ \notag &
		+ \sum_{i < j} \frac{(-\rho b + i - 2)(-\rho b + j - 2)}{2 \cdot 2} \Bigg).
	\end{align}
	We clearly see that terms linear and quadratic $\epsilon$ are polynomial in $m$ and $\rho$.
	This follows from the fact that sums of the form
	\begin{align}\label{eq:harmonicsums}
		\sum_{i=1}^m i^k
	\end{align}
	are polynomials in $m$ of degree $k+1$, and from our ability to rewrite
	\begin{align}
		\sum_{i<j} f(i) f(j) = \frac{1}{2} \left [\left(\sum_{i=1}^m f(i)\right)^2  - \sum_{i=1}^m f(i)^2 \right ].
	\end{align}
	The coefficient in front of $\epsilon^k$ in the whole expression is some finite sum
	(the number of summands depends only on $k$ and not on $m$ or $\rho$) of similar expressions
	and hence analogously is also polynomial in $m$ and $\rho$.
	Since expressions \eqref{eq:harmonicsums} have degree $k+1$ in $m$ and we see that in the coefficient of $\epsilon^k$ the degree of $\rho$ is at most $k$, the total degree of  $Poly_{k}(m, \rho)$ is at most $2k+1$.
\end{proof}

Finally, let's consider the expansion of the $q$-hypergeometric function. It's not simply polynomial,
but instead it is a linear combination of the usual hypergeometric function and its first derivative with polynomial coefficients.

\begin{lemma}\label{lem:hyperpoly}
	For $q=1+\epsilon$,
	\begin{align} 
		& [\epsilon^k] \thephi \left(q^{-m}, q^{\rho b}; q^{\rho b + 1 - m} ; q ; a q \right)
		\\ \notag
		& = Poly_{1,k}(m, \rho) \theF(-m, \rho b; \rho b + 1 - m)(a) \\ \notag
		& + Poly_{2,k}(m, \rho) \frac{d}{d a} \theF(-m, \rho b; \rho b + 1 - m)(a),
	\end{align}
	where  $Poly_{1,k}(m, \rho)$ and $Poly_{1,k}(m, \rho)$ are polynomials in $m$ and $\rho$ of total degree no greater than $5k$.
\end{lemma}

\begin{proof}
	Consider the coefficient in front of $z^n$ in the definition of the $q$-hypergeometric function \eqref{eq:q-hypergeom-def}.
	Analogously to Lemma \ref{lem:q-poch-poly} its coefficient in front of $\epsilon^k$
	is the ratio of the ordinary Pochhammer symbols times a polynomial in $n$, $m$ and $\rho$:
	\begin{align} \notag
		\frac{(-m)_n (\rho b)_n}{(\rho b + 1 - m)_n n!} Poly_k(n, m, \rho),
	\end{align}
	Here $Poly_k(n, m, \rho)$ has total degree in $m$ and $\rho$ no greater than $k$, and its degree in $n$ is no greater than $k+1$.
	
	If we pretend for a moment that the argument of our $q$-hypergeometric function is $a$
	and not $a q$, we would immediately conclude that
	\begin{align} \label{eq:phi-diff-poly}
	&
		[\epsilon^k] \thephi \left(q^{-m}, q^{\rho b}; q^{\rho b + 1 - m} ; q ; a \right)
		\\ \notag &
		= Poly_k(z \frac{d}{d z}, m, \rho) \theF(-m, \rho b; \rho b + 1 - m)(z) \Big{|}_{z = a}.
	\end{align}
	
	The hypergeometric function $\theF (a, b; c)(z)$ satisfies the hypergeometric equation
	\begin{align}
		z(1-z) \frac{d^2}{d z^2} F + \left[ c - (a+b+1)z \right] \frac{d}{dz} F - a b F = 0,
	\end{align}
	which in our case takes the form
	\begin{align}
		a(1-a) \frac{d^2}{d a^2} F + (\rho b + 1 - m)(1 - a) \frac{d}{d a} F + m \rho b F = 0.
	\end{align}
	So we see that we can eliminate all higher derivatives from the formula~\eqref{eq:phi-diff-poly}. We apply this relation no more than $k$ times (as the degree of $Poly_k(n, m, \rho)$ in $n$ is no greater than $k+1$), and with each application the total degree in $m$ and $\rho$ rises by no more than $2$. Thus we see that the resulting polynomials in the coefficients in front of the hypergeometric functions have total degrees in $m$ and $\rho$ no greater than $5k$.
	
	One final observation is that $a q$ instead of $a$ in the last argument of $\thephi$ does not spoil anything. It just results in a finite, $m$-, $n$- and $\rho$-independent resummation of the polynomials	$Poly_k(n, m, \rho)$, which preserves the bound on the total degree. This completes the proof of the lemma.
\end{proof}

Recall the notation $J_m$ introduced in \eqref{eq:Jdef}. Let us prove the following
\begin{proposition}\label{prop:hyperjacrel} We have:
	\begin{align}\label{eq:hyperfirstJ1}
		\frac{\Gamma(m - \rho b)}{\Gamma(m+1) \Gamma(-\rho b)}\theF(-m, \rho b; \rho b + 1 - m)(a) &= (-1)^m \dfrac{(1-a)\rho b}{m} J_{m-1}\\ \label{eq:hyperfirstJ2}
		\frac{\Gamma(m - \rho b)}{\Gamma(m+1) \Gamma(-\rho b)}\frac{d}{d a} \theF(-m, \rho b; \rho b + 1 - m)(a) &= (-1)^{m+1}\rho b \left(J_{m-1} + J_{m-2}\right),
	\end{align}
\end{proposition}
\begin{proof}
	Let us start with proving \eqref{eq:hyperfirstJ1}.
	First note that
	\begin{equation}\label{eq:gammarel}
	\frac{\Gamma(m - \rho b)}{\Gamma(-\rho b)} = (-1)^m \frac{\Gamma(\rho b+1)}{\Gamma(\rho b-m+1)}.
	\end{equation}
	Then note that from the definition of Jacobi polynomials $\jp{m}{\al}{\be}$ we have:
	\begin{equation}\label{eq:hyperjac}
	\frac{\Gamma(\rho b+1)}{\Gamma(m+1)\Gamma(\rho b-m+1)}\theF(-m, \rho b; \rho b + 1 - m)(a) = \jp{m}{\rho b -m}{-1}(1- 2a).
	\end{equation}
	Recall the following well-known relations for the Jacobi polynomials (see e.g. \cite{WTF}):
	\begin{align}	\label{eq:jacnrel1}
		\jp{n}{\al}{\be-1}(x) &= \jp{n-1}{\al}{\be}(x) + \jp{n}{\al-1}{\be}(x), \\ \label{eq:jacnrel2}
		(n+\al+1)\jp{n}{\al}{\be}(x) &= (n+1)\jp{n+1}{\al}{\be}(x) + \left(n+\dfrac{\al}{2}+\dfrac{\be}{2}+1\right)(1-x)\jp{n}{\al+1}{\be}(x).
	\end{align}
	Also recall the three-term relation \eqref{eq:JacobiThreeTerm} proved in Proposition \ref{prop:JacobiThreeTerm}. With the help of the aforementioned three relations we obtain (in what follows we omit the arguments of the Jacobi polynomials for brevity, they are always equal to $1-2a$):
	\begin{align}\label{eq:jactransf1}
		\jp{m}{\rho b -m}{-1} & \mathop{=}_{\eqref{eq:jacnrel1}} \jp{m-1}{\rho b-m}{0}+\jp{m}{\rho b-m-1}{0}
		\\ \notag &
		\mathop{=}_{\eqref{eq:jacnrel1}} \jp{m-2}{\rho b-m}{1}+2\jp{m-1}{\rho b-m-1}{1}+\jp{m}{\rho b-m-2}{1} \\ \nonumber
		&\mathop{=}_{\eqref{eq:jacnrel2}}\dfrac{1}{\rho b -1}\big((m+1)J_{m+1}+(2m+(m+\rho b +1)a J_m\\ \nonumber
		&\phantom{ \mathop{=}_{\eqref{eq:jacnrel2}} } + (m-1+2(m+\rho b)a)J_{m-1}+(m+\rho b -1)a J_{m-2})\big) \\ \nonumber
		&\mathop{=}_{\eqref{eq:JacobiThreeTerm}} \dfrac{(1-a)\rho b}{m} J_{m-1}.
	\end{align}
	Combining \eqref{eq:gammarel}, \eqref{eq:hyperjac} and \eqref{eq:jactransf1} we obtain the proof of \eqref{eq:hyperfirstJ1}.
	
	Let us continue to the proof of \eqref{eq:hyperfirstJ2}. Recall the formula for the derivative of the hypergeometric function:
	\begin{align}
	& \frac{d}{d a} \theF(-m, \rho b; \rho b + 1 - m)(a) 
	\\ \notag &
	= -\dfrac{m\rho b}{\rho b +1 -m}\theF(-m+1, \rho b+1; \rho b + 2 - m)(a)
	\end{align}
	Thus (also recalling the definition of the Jacobi polynomials)
	\begin{align}\label{eq:hyperderivjac}
		&\dfrac{\Gamma(\rho b+1)}{\Gamma(m+1)\Gamma(\rho b-m+1)}\frac{d}{d a}\theF(-m, \rho b; \rho b + 1 - m)(a) \\ \nonumber
		&= -\rho b \dfrac{\Gamma(\rho b+1)}{\Gamma(m)\Gamma(\rho b-m+2)}\theF(-m+1, \rho b+1; \rho b + 2 - m)(a) \\ \nonumber
		&= -\rho b \jp{m-1}{\rho b -m+1}{0}
	\end{align}
	With the help of \eqref{eq:jacnrel1} we get
	\begin{equation}\label{eq:jacrel1appl}
	\jp{m-1}{\rho b -m+1}{0} = \jp{m-2}{\rho b -m+1}{1}+\jp{m-1}{\rho b -m}{1} = J_{m-1}+J_{m-2}
	\end{equation}
	Combining \eqref{eq:gammarel}, \eqref{eq:hyperderivjac}, and \eqref{eq:jacrel1appl} we obtain the proof of \eqref{eq:hyperfirstJ2}.
\end{proof}

Now we are finally ready to prove Proposition \ref{prop:rationality}.
\begin{proof}[Proof of Proposition \ref{prop:rationality}]
	Proposition \ref{prop:rationality} is a straightforward implication of the results of Proposition~\ref{prop:hypergeom}, Lemmas~\ref{lem:coefficients-h-e}--\ref{lem:hyperpoly}, and Proposition~\ref{prop:hyperjacrel}. 
\end{proof}

\subsection{Properties of the $G$-polynomials}
In the following sections we will need to know certain properties of the polynomials $G^1_k(\rho,m)$ and $G^2_k(\rho,m)$ introduced in Proposition \ref{prop:rationality} of the previous subsection. Let us describe and prove these properties.
\begin{lemma}\label{lem:G0} We have:
	\begin{align} \label{eq:Gk10}
		G_k^1(\rho,0)&=\delta_{k,0}(1-a)b\\ \label{eq:Gk20}
		G_k^2(\rho,0)&=0
	\end{align}
\end{lemma}
\begin{proof}
	Note that the pole $m^{-1}$ in the RHS of \eqref{eq:expJacobi} is coming only from \eqref{eq:hyperfirstJ1} and not from \eqref{eq:hyperfirstJ2}, while, at the same time, the only contribution to $G_k^2$ is coming from \eqref{eq:hyperfirstJ2}. Thus, $G_k^2(\rho,m)$ is proportional to $m$, and thus $G_k^2(\rho,0)$ vanishes.
	
	Note that the LHS of \eqref{eq:hyperfirstJ1} makes perfect sense for $m=0$ and is actually equal to $1$ (since, as it follows straightforwardly from the definition, $\theF(0,\al,\be;x)=1$ for any $\al$, $\be$, and $x$). This implies that we have
	\begin{equation}
	\mathop{\lim}_{m\rightarrow 0} \dfrac{J_{m-1}}{m} = \dfrac{1}{\rho b (1-a)}.
	\end{equation}
	In fact, the proof of Proposition \ref{prop:rationality} works perfectly well for $m=0$ if one replaces 
	\begin{align}
	\dfrac{J_{m-1}}{m} \quad \text{with} \quad \dfrac{1}{\rho b (1-a)}
	\end{align}
	and puts $G_k^2(\rho,0) = 0$  in the RHS of \eqref{eq:expJacobi}. It is obvious that for $m=0$ the LHS of \eqref{eq:expJacobi} is equal to $\delta_{k,0}$. Thus we obtain \eqref{eq:Gk10}.
\end{proof}

\begin{lemma}\label{lem:Gdoublezero}
	$\forall k \geq 1$ polynomials $G_k^1(m,m)$ and $G_k^2(m,m)$ are divisible by $m^2$, i.e. they have a double zero at $m=0$.
\end{lemma}
\begin{proof}
	Let $m\in\mathbb{Z}_{> 1}$ in what follows. Consider the LHS of \eqref{eq:expJacobi}. Let us introduce a piece of notation:
	\begin{equation}
	S(z):=\zeta(z)/z.
	\end{equation}
	Note that the Taylor series of $S(z)$ at $z=0$ starts from $1$ and contains only even powers of $z$. With the help of this notation we rewrite the  LHS of \eqref{eq:expJacobi} as follows:
	\begin{align}
		& [u^{2k}w^m]\exp\left(\sum_{i=1}^\infty \frac{a^i-1}{i}\cdot w^i \dfrac{\zeta(iu\rho)}{\zeta(iu b^{-1})}\right)
		\\ \notag &
		=[u^{2k}w^m]\exp\left(\sum_{i=1}^\infty \frac{a^i-1}{i}\cdot w^i \rho b \dfrac{S(iu\rho)}{S(iu b^{-1})}\right)
	\end{align}
	Let us denote the Taylor coefficients of $S(z)^{-1}$ as $s_{2k}$, i.e.
	\begin{equation}
	\dfrac{1}{S(z)} = 1 + \sum_{k=1}^\infty s_{2k}z^{2k}
	\end{equation}
	Then we have 
	\begin{align}\label{eq:jelexp}
		&[u^{2k}w^m]\exp\left(\sum_{i=1}^\infty \frac{a^i-1}{i}\cdot w^i \dfrac{\zeta(iu\rho)}{\zeta(iu b^{-1})}\right)
		\\ \notag &
		=[u^{2k}w^m]\exp\left(\sum_{i=1}^\infty \frac{a^i-1}{i}\cdot w^i \rho b \dfrac{1+O(\rho^2)}{S(iu b^{-1})}\right) \\ \nonumber
		&=\dfrac{a^m-1}{m}\rho b\;[u^{2k}]\dfrac{1}{S(mu b^{-1})} + O(\rho^2)
		\\ \notag &
		 = (a^{m}-1)\be^{1-2k}\,s_{2k}\,m^{2k-1}\,\rho + O(\rho^2).
	\end{align}
	Now consider the RHS of \eqref{eq:expJacobi}. From Proposition \ref{prop:j12exprs} we have:
	\begin{align}\label{eq:jerexp}
		&(-1)^m \dfrac{\rho}{m}\left(G_k^1(\rho,m)J_{m-1}+G_k^2(\rho,m)J_{m-2}\right)\\ \nonumber
		&=G_k^1(\rho,m)\dfrac{1}{\be(1-a)}[w^m]\exp\left(\sum_{i=1}^\infty \frac{a^i-1}{i}\cdot w^i \rho b\right) \\ \nonumber
		&\phantom{ =\ }-G_k^2(\rho,m)\dfrac{1}{\be(1-a)}\dfrac{m-1}{m}[w^{m-1}]\exp\left(\sum_{i=1}^\infty \frac{a^i-1}{i}\cdot w^i \rho b\right) \\ \nonumber
		&=G_k^1(\rho,m)\dfrac{1}{\be(1-a)}\left(\dfrac{a^m-1}{m}\rho b + O(\rho^2)\right) \\ \nonumber
		&\phantom{=\ }-G_k^2(\rho,m)\dfrac{1}{\be(1-a)}\dfrac{m-1}{m}\left(\dfrac{a^{m-1}-1}{m-1}\rho b + O(\rho^2)\right)\\ \nonumber
		&=\dfrac{1}{m(1-a)}\Big(G_k^1(\rho,m)\left(a^{m}-1\right)-G_k^2(\rho,m)\left(a^{m-1}-1\right)\Big)\rho + O(\rho^2).
	\end{align}
	From \eqref{eq:expJacobi} we know that the RHS of \eqref{eq:jelexp} and the RHS of \eqref{eq:jerexp} must coincide, i.e. for any $ m \in \mathbb{Z}_{>1}$ we have
	\begin{equation}\label{eq:g1g2rel}
	G_k^1(\rho,m)\left(a^{m}-1\right)-G_k^2(\rho,m)\left(a^{m-1}-1\right) = (1-a) \be^{1-2k} s_{2k}\, m^{2k} (a^m-1) + O(\rho)
	\end{equation}
	and thus for all $ m \in \mathbb{Z}_{>1}$ we have
	\begin{equation}\label{eq:g1g2rel2}
	G_k^1(0,m)\left(a^{m}-1\right)-G_k^2(0,m)\left(a^{m-1}-1\right) = (1-a) \be^{1-2k} s_{2k}\, m^{2k} (a^m-1).
	\end{equation}
	Here $s_{2k}$ are fixed constants that do not depend on $a$, and $G_k^1(0,m)$ and $G_k^2(0,m)$ are fixed polynomials in $m$ with coefficients depending on $a$.
	
	Since functions $a^{m-1}-1$ and $a^m-1$ are linearly independent over the ring of polynomials in $m$ with coefficients depending on $a$, the fact that we have equality \eqref{eq:g1g2rel2} for all $ m \in \mathbb{Z}_{>1}$ implies that
	\begin{equation} \label{eq:g2rho}
	G_k^2(0,m) = 0.
	\end{equation}
	Furthermore, we obtain that
	\begin{equation} \label{eq:g1div}
	G_k^1(0,m) = m^{2k} \widecheck{G}_k^1(m),
	\end{equation} 
	where $\widecheck{G}_k^1(m)$ is some polynomial in $m$.
	
	Note that Lemma~\ref{lem:G0} implies that $G_k^1(\rho,m)$ and $G_k^2(\rho,m)$ are divisible by $m$. Equality \eqref{eq:g2rho} furthermore implies that $G_k^2(\rho,m)$ is divisible by $\rho$, which means that $G_k^2(m,m)$ is divisible by $m^2$, which is exactly what we wanted to prove regarding $G_k^2(m,m)$.
	
	Now for $G_k^1(\rho,m)$ we have
	\begin{equation}
	G_k^1(\rho,m) = G_k^1(0,m) + \rho \widetilde{G_k}^1(\rho,m),
	\end{equation}
	where $\widetilde{G_k}^1(\rho,m)$ is some polynomial in $\rho$ and $m$. Since $G_k^1(\rho,m)$ is divisible by $m$ as noted above,  $\widetilde{G_k}^1(\rho,m)$ is divisible by $m$ as well. This implies that the second term in the RHS of the equality
	\begin{equation}\label{eq:g1rexpr}
	G_k^1(m,m) = G_k^1(0,m) + m \widetilde{G_k}^1(m,m),
	\end{equation}
	is divisible by $m^2$. Equality \eqref{eq:g1div} implies that for $k\in \mathbb{Z}_{\geq 1}$ the first term in the RHS of  \eqref{eq:g1rexpr} is divisible by $m^2$ as well. Thus, for $k\in \mathbb{Z}_{\geq 1}$ polynomial $G_k^1(m,m)$ is divisible by $m^2$.
\end{proof}

\section{The $\mathcal{A}$-operators and their rationality} \label{sec:aops}
In Section 2 we defined the extended Ooguri-Vafa partition function $\Ze(\pb)$ as a fermionic correlator, whose coefficient in front of the monomial $\pb_{\mu_1/Q}\dots \pb_{\mu_k/Q}$ is proportional to $K_\mu(u)$~\eqref{eq:kmu-def}. In this section we represent $K_\mu(u)$ as a correlator of a product of some operators $\At(\mu_i,u\mu_i)$ and prove that
this correlator is a polylinear combination of
$\xi^1_{\mu_i}$ and $\xi^2_{\mu_i}$ (see Notation \ref{not:xi-coeffs}) with coefficients
that are rational functions in $\mu_i$. 

\subsection{$\mathcal{A}$-operators}
From \eqref{eq:kmu-def} we have
\begin{align} \label{eq:kexpr}
	& K_\mu(u) = 	
	\left\langle
	\exp \lb \sum_{i=1}^\infty \frac{\alpha_{i} p^*_i}{i} \rb
	\exp \lb u \Fc_2 \rb
	\prod_{i=1}^k \frac{\alpha_{-\mu_i}}{\mu_i} A^{b \mu_i}
	\right\rangle\\
	&=
	\left\langle
	\prod_{i=1}^k
	\exp \lb \sum_{i=1}^\infty \frac{\alpha_{i} p^*_i}{i} \rb
	\exp \lb u \Fc_2 \rb
	\frac{\alpha_{-\mu_i}}{\mu_i} A^{b \mu_i}
	\exp \lb - u \Fc_2 \rb 
	\exp \lb -\sum_{i=1}^\infty \frac{\alpha_{i} p^*_i}{i} \rb 
	\right\rangle \nonumber.
\end{align}
Note that
\begin{equation}
e^{u\mathcal{F}_2} \alpha_{-\mu} e^{-u\mathcal{F}_2} = \mathcal{E}_{-\mu} (u\mu)
\end{equation}
and
\begin{equation}
[\alpha_r,\mathcal{E}_{kr-\mu}(u\mu)] = \zeta(ru\mu)\mathcal{E}_{(k+1)r-\mu}(u\mu).
\end{equation}

\begin{definition}\label{def:a-oper} We define the operator $\At(m,u m)$ by the following formula
	\begin{align} \label{eq:atildedef}
		\At(m,u m) :=  \dfrac{A^{bm}}{m}\sum_{k=0}^\infty\Ec_{k-m}(u m)\sum_{\la \vdash k}\prod_{i=1}^{l(\la)}
		\dfrac{
			\lb  \dfrac{p^*_i}{i}\, \zeta \big( i u m \big)\rb^{\la_i-\la_{i+1}} 
		}{(\la_i-\la_{i+1})!}.
	\end{align}
\end{definition}

It is straightforward to see that
	\begin{align} \label{Kcorr}
		K_\mu(u) =  \left\langle \prod_{i=1}^k \At(\mu_i,u\mu_i)\right\rangle.
	\end{align}


Note that the definition \eqref{eq:atildedef} of the $\At$-operators, together with the definition \eqref{eq:curlyEdef} for the $\mathcal{E}$-operators, implies the following structure of the $\At$-operators:
\begin{equation}\label{eq:Acoefssum}
\At(m,um) = \ac_0(m,u)  \Id + \sum_{l\in\mathbb{Z}+1/2}\;\sum_{s\in\mathbb{Z}}\ac_{l,s}(m,u) E_{l-s,l},
\end{equation} 
where $\ac_0(m,u)$ and $\ac_{1,s}(m,u)$ are certain coefficients depending on $m$ and $u$, $\Id$ is the identity operator and $E_{i,j}$ are the operators defined by formula \eqref{eq:Eoperdef}.

\subsection{Rationality of $\mathcal{A}$-operators}

The main goal of this section is to prove the following statement (\emph{the rationality of $\mathcal{A}$-operators}): 
\begin{proposition} \label{prop:aop_rationality}
	For the coefficients $\ac_0(m,u)$ and $\ac_{l,s}(m,u)$ in formula \eqref{eq:Acoefssum} for the $\At$-operators we have:
	\begin{align} \label{eq:aop_matel_form}
	[u^{k}]\ac_{l,s}&=\dfrac{F_{k,l}^1(m,s)\xi^1_m+F_{k,l}^2(m,s)\xi^2_m}{(m+1)\dots(m+s)},  &s\geq 0,\\\label{eq:aop_matel_form_neg}
	[u^{k}]\ac_{l,s}&=\dfrac{\widetilde{F}_{k,l}^1(m,s)\xi^1_m+\widetilde{F}_{k,l}^2(m,s)\xi^2_m}{(m-1)(m-2)\dots(m+s+1)(m+s)^2},  &s< 0,
	\\ \label{eq:fka0}
	[u^{k}]\ac_0&=\dfrac{\check{F}_{k}^1(m)\xi^1_m+\check{F}_{k}^2(m)\xi^2_m}{m^2},  &
	\end{align}
	where $F^1_{k,l}$, $F^2_{k,l}$, $\widetilde{F}^1_{k,l}$, $\widetilde{F}^2_{k,l}$,
	$\check{F}^1_{k}$, $\check{F}^2_{k}$
	are, for fixed $s$, some polynomials in $m$.
	
	Moreover, the degrees in $m$ of $F^1_{k,l}(m,s)$, $F^2_{k,l}(m,s)$, $\widetilde{F}^1_{k,l}(m,s)$, and $\widetilde{F}^2_{k,l}(m,s)$ are no greater than $9k+2+s$ and the degrees of $\check{F}^1_{k}(m)$ and $\check{F}^2_{k}(m)$ are no greater than $9k+2$.
\end{proposition}

\begin{proof}
	
	Let us recall a few facts. 
%
From \eqref{eq:atildedef} and \eqref{eq:curlyEdef} we have
\begin{align} \label{eq:aop_matrixel}
[u^{k}]\ac_{l,s} 
&=[u^{k}] \dfrac{A^{(b-1)m}}{m}e^{um(l+s/2)}\sum_{\la \vdash m+s}\prod_{i=1}^{l(\la)}
 \dfrac{
\lb  \dfrac{a^i-1}{i}\,\dfrac{ \zeta \big( i u m \big)}{\zeta\left(iu b^{-1}\right)}\rb^{\la_i-\la_{i+1}}
}{(\la_i-\la_{i+1})!}  \\ \nonumber
&=\left([u^{k}w^{s+m}]\dfrac{A^{(b-1)m}}{m}e^{um(l+s/2)}\exp\left(\sum_{i=1}^\infty \frac{a^i-1}{i}\cdot w^i \dfrac{\zeta(iu\rho)}{\zeta(iu b^{-1})}\right)\right)\Bigg|_{\rho=m};\\ \label{eq:aop_idel}
[u^{k}]\ac_{0} 
&=[u^{k}] \dfrac{1}{m}\dfrac{A^{(b-1)m}}{\zeta(um)}\sum_{\la \vdash m}\prod_{i=1}^{l(\la)}
 \dfrac{
	\lb  \dfrac{a^i-1}{i}\,\dfrac{ \zeta \big( i u m \big)}{\zeta\left(iu b^{-1}\right)}\rb^{\la_i-\la_{i+1}} 
}{(\la_i-\la_{i+1})!} \\ \nonumber
&=\left([u^{k}w^{m}]\dfrac{1}{m}\dfrac{A^{(b-1)m}}{\zeta(um)}\exp\left(\sum_{i=1}^\infty \frac{a^i-1}{i}\cdot w^i \dfrac{\zeta(iu\rho)}{\zeta(iu b^{-1})}\right)\right)\Bigg|_{\rho=m}.
\end{align}
%
%
From Section \ref{sec:xiJac}, for $\xi$'s we have
%
		\begin{align}\label{eq:xiJrel}
		\xi^1_m &= (-1)^m A^{(b-1)m} J_{m-1}\big|_{\rho=m};\\ \nonumber
		\xi^2_m &= (-1)^m A^{(b-1)m} J_{m-2}\big|_{\rho=m}.
		\end{align}
	
	Recall equation \eqref{eq:aop_matrixel}. Let us drop the factor $e^{um(l+s/2)}$ as it clearly does not affect the statement. From Proposition \ref{prop:rationality} we have:
	\begin{align}\label{eq:aop_coefs}
	&\left([u^{2k}w^{s+m}]\dfrac{1}{m}\exp\left(\sum_{i=1}^\infty \frac{a^i-1}{i}\cdot w^i \dfrac{\zeta(iu\rho)}{\zeta(iu b^{-1})}\right)\right)\Bigg|_{\rho=m} \\ \nonumber
	&\phantom{aa}=(-1)^{m+s}\left(\dfrac{1}{m}\dfrac{\rho}{m+s}\left(G_k^1(\rho,m+s)J_{m+s-1}+G_k^2(\rho,m+s)J_{m+s-2}\right)\right)\bigg|_{\rho=m},
	\end{align}
	where $G_{k}^1$ and $G_k^2$ are polynomial in both arguments.
	
	Now recall the three-term relation for Jacobi polynomials \eqref{eq:JacobiThreeTerm}. It implies that
	\begin{equation}\label{eq:three-term-expr}
	J_k=-\left(a+1+(a-1)\dfrac{\rho b}{k}\right)J_{k-1}-aJ_{k-2}
	\end{equation}
	With the help of this relation we get:
	\begin{align}
	&G_k^1(\rho,m+s)J_{m+s-1}+G_k^2(\rho,m+s)J_{m+s-2}\\ \nonumber
	&=\dfrac{(m+s-1)G_k^2(\rho,m+s) - \left((m+s-1)(a+1)+(a-1)\rho b\right)G_k^1(\rho,m+s)}{m+s-1} J_{m+s-2}\\ \nonumber
	&- aG_k^1(\rho,m+s) J_{m+s-3}
	\end{align}
	If we apply the three-term relation repeatedly $s-1$ more times, we produce an expression precisely of the form
	\eqref{eq:aop_matel_form} (where we should still substitute $m$ for $\rho$ and take into account the polynomiality of $G^1_k$ and $G_k^2$ as well as the relations \eqref{eq:xiJrel} between $\xi$'s and $J$'s). Note that at the last step we acquire a factor $m$ in the denominator, but it gets canceled once we substitute $\rho$ for $m$. Indeed, at the last step $k=m$ in the three-term relation and thus the $\rho/k=\rho/m$ part cancels once we substitute $\rho$ with $m$. Note that the degrees of the resulting polynomials in the numerator are no greater than the degrees of polynomials $G_{k}^1$ and $G_k^2$ increased by $s$, since each application of the three-term relation potentially increases the degree by $1$.
	
	The proof for the case of negative $s$ is completely analogous, we just have to use the three-term relation in different direction expressing $J_{k-2}$ through $J_{k-1}$ and $J_k$.
	The proof for the case of $\ac_0$ is analogous as well. 
\end{proof}

\section{Polynomiality of connected correlators of $\mathcal{A}$-operators} \label{sec:Apoly}

The main goal of the present section is to prove the following
\begin{theorem} \label{prop:Apolynom}	
	Coefficients in $u$-expansions of stable connected correlators of $\At$-operators can be expressed as
	\begin{equation}\label{eq:Acorru}
	[u^k]\langle \At(\mu_1,u)\dots\At(\mu_n,u)\rangle^{\circ} = 	
	\sum_{\substack{(\eta_1,\dots,\eta_n)\\\in\{1,2\}^{n}}}P_{k;\eta_1,\dots,\eta_n}(\mu_1,\dots,\mu_n)\xi^{\eta_1}_{\mu_1}\cdots \xi^{\eta_n}_{\mu_n},
	\end{equation}
	where $P_{k;\eta_1,\dots,\eta_n}(\mu_1,\dots,\mu_n)$ are some polynomials in $\mu_1,\dots,\mu_n$. Stable means $(n,k)\notin \{(1,-1),(2,0)\}$.
\end{theorem}

\begin{remark}
There is an equivalent reformulation of this theorem in terms of expansions on the spectral curve given in Theorem~\ref{teo:functions-on-the-curve}, and these two theorems together are the main theorems of the present paper.	
\end{remark}

Let us introduce the following definition:


\begin{definition}\label{def:aops}
	We define the $\mathcal{A}$-operators as follows (here $m\in \mathbb{Z}_{\geq 1}$):
	\begin{align}
	\mathcal{A}^1(m,u) &:= \sum_{l\in\mathbb{Z}+\frac{1}{2}} \; \sum_{s=-m}^{-1} \; \sum_{k=0}^{\infty} u^k \dfrac{\widetilde{F}_{k,l}^1(m,s)}{(m-1)\dots(m+s+1)(m+s)^2} E_{l-s,l} \\ \nonumber
	&\phantom{:=\ }+ \sum_{l\in\mathbb{Z}+\frac{1}{2}} \; \sum_{s=0}^{\infty} \; \sum_{k=0}^{\infty} u^k \dfrac{F_{k,l}^1(m,s)}{(m+1)\dots(m+s)} E_{l-s,l} \\ \nonumber
		&\phantom{:=\ }+
		\sum_{k=-1}^{\infty} u^k \dfrac{\check{F}_{k}^1(m)}{m^2} \Id,\\
	\mathcal{A}^2(m,u) &:= \sum_{l\in\mathbb{Z}+\frac{1}{2}} \; \sum_{s=-m}^{-1} \; \sum_{k=0}^{\infty} u^k \dfrac{\widetilde{F}_{k,l}^2(m,s)}{(m-1)\dots(m+s+1)(m+s)^2} E_{l-s,l}\\  \nonumber
	&\phantom{:=\ } + \sum_{l\in\mathbb{Z}+\frac{1}{2}} \; \sum_{s=0}^{\infty} \; \sum_{k=0}^{\infty} u^k \dfrac{F_{k,l}^2(m,s)}{(m+1)\dots(m+s)} E_{l-s,l}\\ \nonumber
	&\phantom{:=\ } +
	\sum_{k=-1}^{\infty} u^k \dfrac{\check{F}_{k}^2(m)}{m^2} \Id,
	\end{align}
	where $F_{k,l}^i$, $\widetilde{F}_{k,l}^i$, and $\check{F}_{k}^i$ are the polynomials introduced in Proposition \ref{prop:aop_rationality}. 
\end{definition}
Note that
\begin{equation}
\mathcal{A}^1(m,u)\xi^1_m + \mathcal{A}^2(m,u)\xi_m^2 = \widetilde{\mathcal{A}}(m,um)
\end{equation}

From Proposition \ref{prop:aop_rationality} it is clear that the expressions $P_{\eta_1,\dots,\eta_n}(\mu_1,\dots,\mu_n)$ in Theorem~\ref{prop:Apolynom} are actually correlators of $\mathcal{A}$-operators:
\begin{equation}
P_{\eta_1,\dots,\eta_n}(\mu_1,\dots,\mu_n) = \langle \mathcal{A}^{\eta_1}(\mu_1,u)\dots\mathcal{A}^{\eta_n}(\mu_n,u)\rangle^\circ,
\end{equation}
and the statement of the theorem is equivalent to saying that the above correlators are polynomial in $\mu_1,\dots,\mu_n$. In order to prove this theorem, we first have to introduce certain new operators called $\mathcal{A}_+$ in the following definition. These operators are simpler than the $\mathcal{A}$-operators, but the result of their action on the covacuum is the same (we make a precise statement below, in the proof of Theorem~\ref{prop:Apolynom}).

\begin{definition}\label{def:aplusops}
	We define the $\mathcal{A_+}$-operators as follows (here $m\in \mathbb{Z}_{\geq 1}$):
	\begin{align}\label{eq:Aplus1}
	\mathcal{A}_+^1(m,u) &:= \sum_{k=-1}^{\infty} u^k \dfrac{\check{F}_{k}^1(m)}{m^2} \Id+\sum_{l\in\mathbb{Z}+\frac{1}{2}} \; \sum_{s=1}^{\infty} \; \sum_{k=0}^{\infty} u^k \dfrac{F_{k,l}^1(m,s)}{(m+1)\dots(m+s)} E_{l-s,l} \\ \label{eq:Aplus2}
	\mathcal{A}_+^2(m,u) &:= \sum_{k=-1}^{\infty} u^k \dfrac{\check{F}_{k}^2(m)}{m^2} \Id+\sum_{l\in\mathbb{Z}+\frac{1}{2}} \; \sum_{s=1}^{\infty} \; \sum_{k=0}^{\infty} u^k \dfrac{F_{k,l}^2(m,s)}{(m+1)\dots(m+s)} E_{l-s,l}
	\end{align}
\end{definition}

We will need a certain technical lemma regarding the $A_+$-operators.
%
Note that a given coefficient in the $u$-expansion of the coefficient in front of  $E_{l-s,l}$ in the formula for the operator $\mathcal{A}_+^\eta(m,u)$ is a rational function in $m$, and thus it can be extended into the complex plane. The same holds for the coefficient in front of $\Id$. 
Let us consider these extensions and prove the following
\begin{lemma}\label{lem:aatilderel}
	For $r\in\mathbb{Z}_{\geq 1}$ we have
	\begin{align}\label{eq:aatilderel}
	&
	\langle 0|	\mathop{\Res}_{m=-r}\Big(\mathcal{A}_+^\eta(m,u) -\sum_{k=-1}^{\infty} u^k \dfrac{\check{F}_{k}^\eta(m)}{m^2} \Id\Big) 
	\\ \notag &
	= c^\eta(r) \;\langle 0| \left(\widetilde{\mathcal{A}}(-r,u)-\ac_0(-r,u)  \Id\right),
	\end{align}
	where $c^\eta(r)$ is some scalar coefficient, $\eta=1,2$. In other words, the residue at $-r$ of the $\Id$-less part of any of the two $\mathcal{A}_+$-operators coincides with the $\Id$-less part of the $\widetilde{\mathcal{A}}$-operator taken at $m=-r$ up to a scalar factor, under the condition that both are applied to the covacuum.
	
\end{lemma}
\begin{proof}
	Convention: in what follows we assume that the factor of $e^{um(l+s/2)}$ is absent from \eqref{eq:aop_matrixel}. It is easy to see that this does not affect the conclusion.
	
	Recall the proof of Proposition~\ref{prop:aop_rationality}. Let $s\in\mathbb{Z}_{\geq 0}$. Define
	\begin{equation}
	S_n:=\dfrac{1}{n}\left(\begin{array}{cc}
	\rho b(1-a)-n(1+a) & n \\
	-a n & 0
	\end{array}
	\right).
	\end{equation}
	Also define
	\begin{align}
	e_1&:=\left(\begin{array}{c}
	1\\
	0	
	\end{array}\right) ;\\
	e_2&:=\left(\begin{array}{c}
	0\\
	1	
	\end{array}\right) ;\\
	g_k(\rho,p)&:=G^1_k(\rho,p)\,e_1+G^2_k(\rho,p)\,e_2.
	\end{align}
	Now note that from \eqref{eq:aop_coefs} and \eqref{eq:three-term-expr} we get (for $s\geq 0$)
	\begin{align}
	&
	[u^k][E_{l-s,l}]\mathcal{A}_+^\eta(m,u) 
	\\ \notag & 
	=
	(-1)^m\left(\dfrac{\rho}{m(m+s)}\, e_\eta^T\; S_mS_{m+1}\cdots S_{m+s-1}\; g_k(\rho,m+s)\right)\bigg|_{\rho=m}.
	\end{align}
	This follows from the fact that the matrix $S_k$ encompasses the three-term relation. Here by $[u^k][E_{l-s,l}]\mathcal{A}_+^\eta(m,u)$ we mean the coefficient in front of $u^k$ in the $u$-expansion of the coefficient in front of the operator $E_{l-s,l}$ in the corresponding formula \eqref{eq:Aplus1} or \eqref{eq:Aplus2}.
	
	In what follows we always consider $s\geq 0$.
	
	Note that from \eqref{eq:Aplus1} and \eqref{eq:Aplus2} it is clear that
	\begin{align}
	&\forall z\in\mathbb{C}\setminus\{-s,-s+1,\dots,-1\} 
	&\mathop{\Res}_{m=z}[u^k][E_{l-s,l}]\mathcal{A}_+^\eta(m,u) =0.
	\end{align}
	In particular, it is easy to see that the residue at $z=0$ vanishes.
	 
	Consider $r\in\{1,2,\dots,s-1\}$. From the form of $S_k$ and from the fact of polynomiality of $g_k$ we have
	\begin{align}\label{eq:asexp}
	&\mathop{\Res}_{m=-r}[u^k][E_{l-s,l}]\mathcal{A}_+^\eta(m,u)\\ \nonumber
	&=(-1)^r\dfrac{1}{s-r}\left( e_\eta^T\; S_{-r}S_{-r+1}\cdots S_{-1} \left(\begin{array}{cc}
	\rho b(1-a) & 0 \\
	0 & 0
	\end{array}
	\right) S_1 \cdots S_{s-r-1}\; g_k(\rho,s-r)\right)\bigg|_{\rho=-r}\\ \nonumber
	&= \left(\rho b (1-a)\; e_\eta^T\; S_{-r}S_{-r+1}\cdots S_{-1}\; e_1\right)\Big|_{\rho=-r}\; \cdot\; \left(\dfrac{(-1)^r}{s-r}\, e_1^T\; S_1S_{2}\cdots S_{s-r-1}\; g_k(\rho,s-r)\right)\bigg|_{\rho=-r}.
	\end{align}
	The first equality followed from the fact that there is a simple pole at $-r$ in $S_{m-r}$, while all the other factors do not have poles at this point.
	
	Now, consider $\widetilde{\mathcal{A}}(-r,u)$. Again, recall \eqref{eq:aop_coefs} and \eqref{eq:three-term-expr}. For $r\in\{1,2,\dots,s-1\}$, we have
	\begin{equation}\label{eq:atildesexp}
	[u^k][E_{l-s,l}]\widetilde{\mathcal{A}}(-r,u) = A^{(b-1)(-r)}\left(\dfrac{(-1)^{s-r}}{s-r}\, e_1^T\; S_1S_{2}\cdots S_{s-r-1}\; g_k(\rho,s-r)\right)\bigg|_{\rho=-r}.
	\end{equation}
	Here we used the fact that $J_0=1$ and $J_{-1}=0$.
	
	Note that the RHS of \eqref{eq:atildesexp} differs from the RHS of \eqref{eq:asexp} only by a factor of
	\begin{equation}\label{eq:ameldifffact}
	A^{(b-1)(-r)} (-1)^r\left(\rho b (1-a)\; e_\eta^T\; S_{-r}S_{-r+1}\cdots S_{-1}\; e_1\right)\Big|_{\rho=-r},
	\end{equation}
	which does not depend on $s$, $l$ or $k$.
	
	Now let us consider the case $r=s$. Recall Lemma~\ref{lem:G0}. It provides an explicit formula for $g_k(\rho,0)$ which allows us to obtain the following:
\begin{align}
&\mathop{\Res}_{m=-s}[u^k][E_{l-s,l}]\mathcal{A}_+^\eta(m,u) 
\\ \notag &
= 
(-1)^{s+1}\left(\dfrac{\rho}{s}\, e_\eta^T\; S_{-s}S_{-s+1}\cdots S_{-1}\; g_k(\rho,0)\right)\bigg|_{\rho=-s}\\ \nonumber
&=
(-1)^{s+1}\left(\dfrac{\rho}{s}\, e_\eta^T\; S_{-s}S_{-s+1}\cdots S_{-1}\; e_1 \delta_{k,0} (1-a) b\right)\bigg|_{\rho=-s}\\ \nonumber
&=(-1)^{s}\left(\rho b (1-a)\; e_\eta^T\; S_{-s}S_{-s+1}\cdots S_{-1}\; e_1\right)\Big|_{\rho=-s} \; \cdot \left(- \dfrac{\delta_{k,0}}{s}\right).
\end{align}
Now consider $[u^k][E_{l-s,l}]\widetilde{\mathcal{A}}(-s,u)$. From \eqref{eq:aop_matrixel} (remember that we drop the factor $e^{um(l+s/2)}$ in that formula; it is easy to see that its inclusion does not affect the conclusion) we immediately obtain (as it is very easy to take the coefficient in front of $w^0$) that
\begin{equation}
[u^k][E_{l-s,l}]\widetilde{\mathcal{A}}(-s,u) = -\dfrac{\delta_{k,0}}{s} A^{(b-1)(-s)}.
\end{equation}
We conclude that for this $r=s$ case expressions $\displaystyle\mathop{\Res}_{m=-s}[u^k][E_{l-s,l}]\mathcal{A}_+^\eta(m,u)$ and $[u^k][E_{l-s,l}]\widetilde{\mathcal{A}}(-s,u)$ differ again by the same factor \eqref{eq:ameldifffact}.

	

This proves the lemma since all parts with $s\leq 0$ vanish when acting on the covacuum.
%
\end{proof}

Now we are ready to prove Theorem \ref{prop:Apolynom}. We partially follow the logic described in \cite{KLS} and~\cite{KLPS}.

\begin{proof}[Proof of Theorem \ref{prop:Apolynom}]
	As mentioned above, Proposition \ref{prop:aop_rationality} together with Definition \ref{def:aops} imply that the connected correlator of $\At$-operators indeed has the desired form:
	\begin{equation}
		[u^k]\langle \At(\mu_1,u)\dots\At(\mu_n,u)\rangle^{\circ}=
	\sum_{\substack{(\eta_1,\dots,\eta_n)\\\in\{1,2\}^{n}}}P_{k;\eta_1,\dots,\eta_n}(\mu_1,\dots,\mu_n)\xi^{\eta_1}_{\mu_1}\cdots \xi^{\eta_n}_{\mu_n},
	\end{equation}
	with
	\begin{equation}\label{eq:acorr1}
	P_{k;\eta_1,\dots,\eta_n}(\mu_1,\dots,\mu_n) = [u^k]\langle \mathcal{A}^{\eta_1}(\mu_1,u)\dots\mathcal{A}^{\eta_n}(\mu_n,u)\rangle^\circ.
	\end{equation}
	We just need to prove that these expressions $P_{k;\eta_1,\dots,\eta_n}(\mu_1,\dots,\mu_n)$ are polynomial in $\mu_1,\dots,\mu_n$.
	
	Note that correlators
	\begin{equation}\label{eq:atcors}
	[u^k]\langle \At(\mu_1,u)\dots\At(\mu_n,u)\rangle^{\circ}
	\end{equation}
	are symmetric in $\mu_1,\dots,\mu_n$.
	Together with the fact that functions $\xi^1_m$ and $\xi^2_m$ (as functions of $m$) are linearly independent over the ring of polynomials in $m$, this means that it is sufficient to prove that the expressions $P_{k;\eta_1,\dots,\eta_n}(\mu_1,\dots,\mu_n)$ are polynomial just in $\mu_1$ with the degree independent of $\mu_2,\dots,\mu_n$, and this will imply the polynomiality in $\mu_2,\dots,\mu_n$.
	
	
	Let us proceed to proving 
	that the correlator in the RHS of \eqref{eq:acorr1} is polynomial in $\mu_1$. We will also prove that its degree in $\mu_1$ is bounded by a certain number which depends only on $k$ and does not depend on $\mu_2,\dots,\mu_n$.
	
	Note that
	\begin{equation}\label{eq:leftvacact}
	\forall s \leq 0 \quad \langle 0 | E_{l-s,l} = 0,
	\end{equation}
	Recall Definition \ref{def:aops}. It can be rewritten as follows:
	\begin{equation}\label{eq:assum}
		\mathcal{A}^\eta(m,u) = \sum_{s=-m}^\infty \mathcal{A}_s^\eta(m,u),
	\end{equation}
	where
	\begin{equation*}	
	\mathcal{A}_s^\eta(m,u) = \left\{ \begin{array}{ll}\sum_{l\in\mathbb{Z}+\frac{1}{2}} \; \sum_{k=0}^{\infty} u^k \dfrac{\widetilde{F}_{k,l}^\eta(m,s)}{(m-1)\dots(m+s+1)(m+s)^2} E_{l-s,l} & -m\leq s \leq -1\\
\sum_{l\in\mathbb{Z}+\frac{1}{2}} \; \sum_{k=0}^{\infty} u^k\left( F_{k,l}^\eta(m,0) E_{l,l} + \dfrac{\check{F}_{k}^\eta(m)}{m^2}\Id\right) & s=0\\ 
	\sum_{l\in\mathbb{Z}+\frac{1}{2}} \; \sum_{k=0}^{\infty} u^k \dfrac{F_{k,l}^\eta(m,s)}{(m+1)\dots(m+s)} E_{l-s,l} & 1\leq s
	\end{array}\right.
	\end{equation*}
	Note that a correlator of the form
	\begin{equation}
	\left\langle \left(\sum_{l_1\in\mathbb{Z}+\frac{1}{2}}c_{l_1,s_1}E_{l_1-s_1,l_1}\right)\cdots \left(\sum_{l_n\in\mathbb{Z}+\frac{1}{2}}c_{l_n,s_n}E_{l_n-s_n,l_n}\right)\right\rangle 
	\end{equation}
	can be nonzero only if $\sum_{i=1}^n s_i = 0$.
	Thus when one writes the $\mathcal{A}$-operators in terms of the $s$-sums \eqref{eq:assum} in the correlator
	\begin{equation*}
	[u^k]\langle \mathcal{A}^{\eta_1}(\mu_1,u)\dots\mathcal{A}^{\eta_n}(\mu_n,u)\rangle
	\end{equation*}
	and expands the brackets, only the terms with  $\sum_{i=1}^n s_i = 0$ will survive. Also note that in any surviving term in this expression we have $s_1\geq 0$, due to \eqref{eq:leftvacact}. Finally, note that from \eqref{eq:assum} we have $s_i\geq -\mu_i$. Thus, for fixed $\mu_2,\dots,\mu_n$ we have
	\begin{equation}
	0\leq s_1 \leq \mu_2+\dots+\mu_n.
	\end{equation} 
	Since 
	 $F_{k,l}^\eta(m,s)$ and $\check{F}_{k}^\eta(m)$ are 
	polynomials, this implies that the correlator
	\begin{equation*}
	[u^k]\langle \mathcal{A}^{\eta_1}(\mu_1,u)\dots\mathcal{A}^{\eta_n}(\mu_n,u)\rangle
	\end{equation*} 
	is a rational function in $\mu_1$ for fixed $\mu_2,\dots,\mu_n$. Moreover, we see that it can have at most simple poles at the negative integer points, and no other poles. The total degree in $\mu_1$ of this correlator as a rational function (i.e. the difference between the degree of the numerator and the denominator) is no greater than $9(k+n-1)+2$. This follows from the estimates on the degrees of polynomials $F_{k,l}^\eta(m,s)$ and $\check{F}_{k}^\eta(m)$ obtained in Proposition \ref{prop:aop_rationality}. We just need to replace $k$ with $k+n-1$ as the $u$-expansions of correlators $\mathcal{A}^{\eta_2}(\mu_2,u),\dots,\mathcal{A}^{\eta_n}(\mu_n,u)$ start with the $u^{-1}$-term.
	
	The fact that the disconnected correlators are rational functions in $\mu_1$ automatically implies that the connected ones are rational in $\mu_1$ as well. In order to prove that stable connected correlators are polynomial in $\mu_1$ we only need to prove that they do not have poles at negative integer points.

	Recall Definition \ref{def:aplusops}. Note that due to \eqref{eq:leftvacact} we have
	\begin{equation}\label{eq:covacAac}
	\langle 0 | 	\mathcal{A}^\eta(m,u) = \langle 0 | 	\mathcal{A}_+^\eta(m,u),
	\end{equation}
	and thus
	\begin{equation}\label{eq:multiAdis}
	\langle \mathcal{A}^{\eta_1}(\mu_1,u)\mathcal{A}^{\eta_2}(\mu_2,u)\dots\mathcal{A}^{\eta_n}(\mu_n,u)\rangle = \langle \mathcal{A}_+^{\eta_1}(\mu_1,u)\mathcal{A}^{\eta_2}(\mu_2,u)\dots\mathcal{A}^{\eta_n}(\mu_n,u)\rangle
	\end{equation}
	and
	\begin{equation}
	\langle \mathcal{A}^{\eta_1}(\mu_1,u)\mathcal{A}^{\eta_2}(\mu_2,u)\dots\mathcal{A}^{\eta_n}(\mu_n,u)\rangle^\circ = \langle \mathcal{A}_+^{\eta_1}(\mu_1,u)\mathcal{A}^{\eta_2}(\mu_2,u)\dots\mathcal{A}^{\eta_n}(\mu_n,u)\rangle^\circ.
	\end{equation}
	
%
	Consider the disconnected correlator in the RHS of \eqref{eq:multiAdis}. Recall that
	\begin{equation}
	\mathcal{A}_+^\eta(m,u) := \sum_{k=-1}^{\infty} u^k \dfrac{\check{F}_{k}^\eta(m)}{m^2} \Id+\sum_{l\in\mathbb{Z}+\frac{1}{2}} \; \sum_{s=1}^{\infty} \; \sum_{k=0}^{\infty} u^k \dfrac{F_{k,l}^\eta(m,s)}{(m+1)\dots(m+s)} E_{l-s,l}.
	\end{equation}
%
%
	Consider the $\Id$-part of $\mathcal{A}^{\eta_1}_+$. Note that the inclusion-exclusion formula for the connected correlator in terms of the disconnected ones for $n\geq 2$ will always contain the term of the form 
	\begin{equation}
	\langle \mathcal{A}_+^{\eta_1}(\mu_1,u)\rangle\langle \mathcal{A}^{\eta_2}(\mu_2,u)\dots\mathcal{A}^{\eta_n}(\mu_n,u)\rangle,
	\end{equation}
	with the opposite sign. Note that in the one-point correlator only the $\Id$-part gives a nonzero contribution. More precisely, 
\begin{equation}
\langle \mathcal{A}_+^{\eta_1}(\mu_1,u)\rangle = \sum_{k=-1}^{\infty} u^k \dfrac{\check{F}_{k}^\eta(m)}{m^2}
\end{equation}
This is precisely the factor which the $\Id$-part of the operator	$\mathcal{A}^{\eta_1}_+$ contributes to $\langle \mathcal{A}_+^{\eta_1}(\mu_1,u) \mathcal{A}^{\eta_2}(\mu_2,u)\dots\mathcal{A}^{\eta_n}(\mu_n,u)\rangle$, i.e.
\begin{align}
&
\left\langle \sum_{k=-1}^{\infty} u^k \dfrac{\check{F}_{k}^\eta(m)}{m^2} \Id \mathcal{A}^{\eta_2}(\mu_2,u)\dots\mathcal{A}^{\eta_n}(\mu_n,u)\right\rangle 
\\ \notag &
= \langle \mathcal{A}_+^{\eta_1}(\mu_1,u)\rangle\langle \mathcal{A}^{\eta_2}(\mu_2,u)\dots\mathcal{A}^{\eta_n}(\mu_n,u)\rangle
\end{align}
This means that these contributions precisely cancel in the inclusion-exclusion formula. Similar reasoning proves that for $n\geq 2$ the $\Id$-parts of the $\mathcal{A}$-operators do not give any nonzero contributions into the connected correlator at all. Thus any connected multi-point correlator can only have poles
coming from the $\Id$-less parts of the $\mathcal{A}$-operator.

Let us return to the correlator of $\At$-operators. We are interested in its dependence on $\mu_1$. Let us prove that  for any $r$ we have
\begin{equation}\label{eq:corrreszero}
\mathop{\Res}_{\mu_1 = - r}\corc{ \mathcal{A}^{\eta_1}(\mu_1,u)\At(\mu_2,u)\cdots\At(\mu_n,u)} = 0.
\end{equation}
For $r\notin \mathbb{Z}_{>0}$ it is clear. 
	Consider $r\in \mathbb{Z}_{>0}$. 
	From Lemma~\ref{lem:aatilderel} for the disconnected correlator we have
	\begin{equation}
	\mathop{\Res}_{\mu_1 = - r}\cord{ \mathcal{A}^{\eta_1}(\mu_1,u)\prod_{i=2}^n \At(\mu_i,u)} = 
	c^{\eta_1}(\mu_1)
	\cord{\widetilde{\mathcal{A}}(-r,u) \prod_{i=2}^n \At(\mu_i,u)},
	\end{equation}
	where \( c^{\eta_1}(\mu_1)\) is the coefficient in Lemma~\ref{lem:aatilderel}.
%
	Recalling equations \eqref{eq:kexpr}, \eqref{Kcorr} and \eqref{eq:atildedef} we can see that the RHS of the previous equality reduces to
	\begin{equation}\label{eq:rescord}
	C\cord{\exp \lb \sum_{i=1}^\infty \frac{\alpha_{i} p^*_i}{i} \rb
		\exp \lb u \Fc_2 \rb
		\alpha_r \prod_{i=2}^n \frac{\alpha_{-\mu_i}}{\mu_i}} 
	\end{equation}
	for some specific coefficient \( C\) that depends only on $r$ and $\eta_1$.
	Because \( [\alpha_k, \alpha_l] = k\delta_{k+l,0} \), and \( \alpha_{r} \) annihilates the vacuum, this residue is zero unless one of the \( \mu_i \) equals \( r \) for \( i \geq 2\).
	
	Now return to the connected $n$-point correlator for $n>2$. It can be calculated from the disconnected one by the inclusion-exclusion principle, so in particular it is a finite sum of products of disconnected correlators. Hence the connected correlator is also a rational function in $\mu_1$, and all possible poles must be inherited from the disconnected correlators.
	The above reasoning implies that we can assume \( \mu_i = r \) for some \( i \geq 2\). Without loss of generality we can assume that it is the case for $i=2$, and $\mu_i\not= r $ for $i\geq 3$.  Then we get a contribution from \eqref{eq:rescord}, but this is canceled in the inclusion-exclusion formula exactly by the term coming from
	\begin{align}
	&
	\mathop{\Res}_{\mu_1 = - r}\cord{\mathcal{A}^{\eta_1}(\mu_1,u)\At(r,u) }\cord{ 
		\prod_{i=3}^n
		\At(u, \mu_j)}
	\\ \nonumber
	&= C\cord{\exp \lb \sum_{i=1}^\infty \frac{\alpha_{i} p^*_i}{i} \rb
		\exp \lb u \Fc_2 \rb
		\alpha_r\, \alpha_{-r}}\cord{\exp \lb \sum_{i=1}^\infty \frac{\alpha_{i} p^*_i}{i} \rb
		\exp \lb u \Fc_2 \rb
		\prod_{i=3}^n \alpha_{-\mu_i}} \\ \nonumber
	&= C \cord{\exp \lb \sum_{i=1}^\infty \frac{\alpha_{i} p^*_i}{i} \rb
		\exp \lb u \Fc_2 \rb
		\alpha_r \prod_{i=2}^n \alpha_{-\mu_i}}
	\end{align}
Thus we have proved \eqref{eq:corrreszero} for $n>2$, which implies that for $n>2$ we have
\begin{equation}\label{eq:acorr1form}
[u^k]\langle \At(\mu_1,u)\dots\At(\mu_n,u)\rangle^{\circ}=
\sum_{\substack{\eta_1\in\{1,2\}}}\widetilde{P}_{k;\eta_1}(\mu_1;\mu_2,\dots,\mu_n)\xi^{\eta_1}_{\mu_1},
\end{equation}
where $\widetilde{P}_{k;\eta_1}(\mu_1;\mu_2,\dots,\mu_n)$ is some expression polynomial in $\mu_1$. Since the total degree of the correlator as a rational function in $\mu_1$ was bounded by $9(k+n-1)+2$ as explained above, now that we know that it is actually a polynomial, we see that its degree of $\widetilde{P}_{k;\eta_1}(\mu_1;\mu_2,\dots,\mu_n)$ in $\mu_1$ is still bounded by the same number, i.e. by $9(k+n-1)+2$.
	
	Now consider the case of $2$-point correlators. Let us prove that they do not have poles at negative integers for the nonzero genus case. For $r\in\mathbb{Z}_{>0}$ we have
\begin{align}
	&\mathop{\Res}_{\mu_1 = - r}[u^k]\cord{\mathcal{A}^{\eta_1}(\mu_1,u)\At(\mu_2,u) }
	\\ \notag &
	= C[u^k]\cord{\exp \lb \sum_{i=1}^\infty \frac{\alpha_{i} p^*_i(u)}{i} \rb
		\exp \lb u \Fc_2 \rb
		\alpha_r\, \alpha_{-\mu_2}}.
\end{align}
	This is nonzero only for $\mu_2=r$ and equal to
	\begin{equation*}
	C[u^k]\cord{\exp \lb \sum_{i=1}^\infty \frac{\alpha_{i} p^*_i(u)}{i} \rb
		\exp \lb u \Fc_2 \rb}.
	\end{equation*}
	Note that a nonzero power of $u$ comes either from $p^*_i(u)$, and then it comes together with some $\alpha_i$ where $i>0$, or from the exponential, and then it comes together with $\mathcal{F}_2$. It cannot come from the exponential, since $\mathcal{F}_2|0\rangle=0$, and it cannot come from $p^*_i(u)$, since $\al_i|0\rangle=0$ for $i>0$. Thus for $k>0$ the residue vanishes. The bound on the degree in $\mu_1$ in this case is an obvious implication of Proposition \ref{prop:aop_rationality}.
	
	What remains is to prove that stable one-point correlators are polynomial. Note that
	\begin{equation}
	[u^k]\left\langle\mathcal{A}(m,u)\right\rangle=\dfrac{\check{F}_{k}^\eta(m)}{m^2},
	\end{equation}
	since for any $i,\,j$ we have $\left\langle E_{i,j}\right\rangle=0$. Note that equations \eqref{eq:fka0}, \eqref{eq:aop_idel}, \eqref{eq:expJacobi}, and \eqref{eq:xiJrel} together imply that for $k\in\mathbb{Z}_{\geq 0}$
	\begin{align}
	\check{F}_{2k}^\eta(m)&=0,\\ \label{eq:fkexpr}
	\check{F}_{2k-1}^\eta(m)&=m^2\sum_{i=0}^k G_{k-i}^\eta(m,m)\,\cdot [u^{2i-1}]\dfrac{1}{\zeta(um)} \dfrac{1}{m},
	\end{align}
	where $G^\eta_k(\rho,m)$ are the polynomials from the RHS of $\eqref{eq:expJacobi}$. We need to prove that polynomials  $\check{F}_{2k-1}^\eta(m)$ are divisible by $m^2$ for $k\in \mathbb{Z}_{\geq 1}$. Note that the only interesting term in the sum in the RHS of  \eqref{eq:fkexpr} is the term with $i=0$ as we know that $G_k^\eta(m,m)$ are polynomials in $m$ and that $[u^{2i-1}]\zeta(um)^{-1}$ is polynomial in $m$ divisible by $m^2$ for $i>0$. Thus we only need to prove that $G_k^\eta(m,m)$ is divisible by $m^2$ for $k>0$. But this is exactly the result of Lemma~\ref{lem:Gdoublezero}.  Hence, we have proved the polynomiality of the stable one-point correlators. The bound on the degree in $\mu_1$ in this case is, again, an obvious implication of Proposition \ref{prop:aop_rationality}.
	
	We conclude that the statement in \eqref{eq:acorr1form} holds for $n=1$ and $n=2$ as well, as long as $(n,k)\notin \{(1,-1),(2,0)\}$.	Together with the fact that the correlator \eqref{eq:atcors} is symmetric in $\mu_1,\dots,\mu_n$ and the fact that functions $\xi^1_m$ and $\xi^2_m$ (as functions of $m$) are linearly independent over the ring of polynomials in $m$, this implies the statement of the theorem.
	
	This completes the proof of the polynomiality of the stable  $\mathcal{A}$-correlators.
\end{proof}

\section{Towards a combinatorial proof of topological recursion} \label{sec:combprooftr}

Recall the connected correlators of the extended Ooguri-Vafa partition function defined in Equation~\eqref{eq:definitionCgmu}. From connected correlators we can define some correlation differentials
$\tilde{\omega}_{g,n}$ via their expansion at $\Lambda = \infty$
\begin{align}
	\tilde\omega_{g,n} &:= Q^{-n} \sum_{\mu_1 \dots \mu_n = 1}^\infty 
	C^{(g)}_{\mu_1 \dots \mu_n} \prod_{i = 1}^n (-\mu_i) \frac{d \Lambda_i}{\Lambda_i^{\mu_i + 1}}.
\end{align}
We \textit{a priori} do not know whether these differentials $\tilde\omega_{g,n}$
coincide with correlation differentials $\omega_{g,n}$, obtained through
topological recursion, or even whether they are well-defined on the curve.
However, Theorem~\ref{thm:xi-functions} implies that the quasipolynomiality result in Theorem~\ref{prop:Apolynom} is equivalent to the following
\begin{theorem} \label{teo:functions-on-the-curve}
	Correlation multidifferentials $\tilde\omega_{g,n}$ are well-defined
	meromorphic differentials on the spectral curve \eqref{eq:bem-curve},
	moreover, they are expressed in terms of $\xi$-functions \eqref{eq:xi-functions-nonperturbative} in the following way
	\begin{align}
		\tilde\omega_{g,n} (U_1,\dots,U_n) = \sum_{\substack{a_1\dots a_n \\ d_1 \dots d_n}} c_{d_1\dots d_n}^{a_1 \dots a_n} \prod_{i=1}^n d\br{\frac{d}{d\lambda_i}}^{d_i}  \xi^{a_i}(U_i),
	\end{align}
	where sum over $d_i$ is finite and each $a_i$ runs from $1$ to $2$.
\end{theorem}

This is another way to state the main result of this paper.

\begin{remark}
	It is straightforward to check that $\tilde\omega_{g,n}$ satisfy
	the linear abstract loop equation, as well as the projection property
	\cite{BoSh15,BorotEynardOrantin} (the necessary checks for unstable correlation differentials are performed in the next section ). So, in order to be able to claim that $\tilde\omega_{g,n} = \omega_{g,n}$, it remains to prove the quadratic abstract loop equation.
\end{remark}

\section{Unstable correlators} \label{sec:unstcorr}
In this section we prove that the unstable correlation differentials for the BEM spectral curve coincide with the expressions derived from the semi-infinite wedge formalism.

\subsection{The case $(g,n) = (0,1)$} We consider the $1$-point correlation differentials in genus $0$.
\begin{theorem}\label{thm:01correlators} We have:
	\begin{align}
		\frac{d}{d\lambda} F_{0,1} (U) \frac{(-1)}{Q^2} + \log(A^2) = y(U) - \frac{\gamma}{Q}x(U)
	\end{align}
\end{theorem}

In order to prove this theorem, we first calculate the $1$-point vacuum expectation in genus $0$.

\begin{lemma}\label{lem:01vacuumcomp} We have
	\begin{align} \label{eq:01vacuum}
		C_{m}^{(0)} = \frac{Q (1-A^2)}{m^2} A^{(b-1)m} (-1)^m \mathcal{P}_{m-1}^{((b-1)m, 1)}(1 - 2 A^2)
	\end{align}
\end{lemma}

\begin{proof}
	Notice that
	\begin{align} \label{eq:C-from-K}
		C_{m}^{(0)} = \frac{1}{b} Q \bs{u^{-1}} K_m (u).
	\end{align}
	Since the only term that contributes to the vacuum expectation in $\mathcal{E}_0$ is the summand proportional to the identity operator, we have:
	\begin{align}
		&  \bs{u^{-1}} K_m (u)  =  \ \bs{u^{-1}} \ba{\At(m, u m)} \\ \notag
		&  = \ \bs{u^{-1}} \dfrac{A^{b m}}{m \zeta(u m) } \sum_{\lambda \vdash m}\prod_{i=1}^{l(\lambda)} 
		\dfrac{    \br{\dfrac{(A^i-A^{-i})}{i}\, m  b \dfrac{\mathcal{S}(i u m)}{\mathcal{S}(i u  b^{-1})}}^{\la_i-\la_{i+1}}
		}{(\lambda_i-\lambda_{i+1})!}
		\\ \notag
		&  = \ \dfrac{A^{b m}}{m^2} \sum_{\lambda \vdash m}\prod_{i=1}^{l(\lambda)}
		\br{\dfrac{1}{(\lambda_i-\lambda_{i+1})!}
			\br{\dfrac{(A^i-A^{-i})}{i}\, m  b}^{\la_i-\la_{i+1}}}
		\\ \notag
		&  =  \ \dfrac{ b (1 - A^2)}{m^2} A^{(b-1)m} (-1)^m \mathcal{P}_{m-1}^{(( b-1)m, 1)}(1 - 2 A^2),
	\end{align}
	where we used Corollary \ref{prop:Jacobiseries} in the last line.
	Plugging this into \eqref{eq:C-from-K} proves the lemma.
\end{proof}

\begin{proof}[Proof of Theorem~\ref{thm:01correlators}]
	Compare the expression for $C^{(0)}_{m}$ given by Lemma~\ref{lem:01vacuumcomp} with expansion of function $\xi^1(U)$ given by Equation~\eqref{eq:xi-functions-expansion}. 
	We see that the genus zero $1$-point free energy $F_{0,1}$ analytically continued to the whole curve must be equal to
	\begin{align}
		F_{0,1}(U) = \sum_{m=1}^\infty C_m^{(0)} \Lambda^{-m} = Q (1-A^2) \br{\dfrac{d}{d\lambda}}^{-2} \xi^1(U)
	\end{align}
	for an appropriate choice of the integration constants (such that the constant and the linear in $\Lambda$ terms do not appear in the expansion of $F_{0,1}$); $\br{d/d\lambda}^{-2}$ appears due to the presence of the factor $1/m^2$ in \eqref{eq:01vacuum}.
	Integrating $\xi^1(U)$ once we get
	\begin{align}
		\br{\dfrac{d}{d\lambda}}^{-1} \xi^1(U) = \frac{(-1)}{1 - A^2}
		\log \br{\dfrac{1 - A^{ b+1} U}{A^2 - A^{ b+1} U}}.
	\end{align}
	Comparing this expression with the formula for $y$ given in Equation~\eqref{eq:bem-curve} proves the theorem
\end{proof}

\subsection{The case $(g,n) = (0,2)$}


In this section we prove that connected 2-point resolvent in genus 0 agrees with the expansion of the standard Bergman kernel on genus zero spectral curve.

\begin{theorem} We have:
	\begin{align} \label{eq:2-point-free-energy}
		F_{0,2}(U_1, U_2) = Q^2 \bs{\log(U_1 - U_2) - \log(\Lambda_1 - \Lambda_2)}.
	\end{align}
\end{theorem}

\begin{proof}
	From explicit form of $\Lambda(U)$ it is easy to conclude that right hand side of Equation~\eqref{eq:2-point-free-energy}
	has expansion at $\Lambda_{1,2} = \infty$ of the form
	\begin{align}
		\sum_{\mu_1,\mu_2 = 1}^\infty c_{\mu_1,\mu_2} \Lambda^{-\mu_1} \Lambda^{-\mu_2}
	\end{align}
	with some coefficients $c_{\mu_1,\mu_2}$.
	Therefore, the theorem follows from Lemmas \ref{lem:conn-2-point} and \ref{lem:2-point-log} below that provide the explicit computation of the expansion of the left hand side and the right hand side of Equation~\eqref{eq:2-point-free-energy}.
\end{proof}

\begin{lemma} \label{lem:conn-2-point} For the connected, genus zero part of the 2-point correlator we have:
	\begin{align}
		C^{(0)}_{\mu_1, \mu_2} = 
		Q^2 \frac{(a - 1)  b}{(\mu_1 + \mu_2)} \left (
		((a + 1) + (a - 1)  b) \xi^1_{\mu_1} \xi^1_{\mu_2}
		+ a \left ( \xi^1_{\mu_1} \xi^2_{\mu_2} + \xi^2_{\mu_1} \xi^1_{\mu_2} \right) \right),
	\end{align}
	where $\xi^1_m$ and $\xi^2_m$ are the coefficients in front of $\Lambda^{-m}$ in the $\Lambda^{-1}$-expansion of the functions $\xi^1(U)$
	and $\xi^2(U)$, respectively.
\end{lemma}

\begin{proof} The connected $2$-point correlator is manifestly given in terms of semi-infinite wedge formalism by
	\begin{align} \label{eq:2-point-in-mu}
		& 
		\Bigg \langle \frac{Q^2 A^{ b(\mu_1+\mu_2)}}{\mu_1\mu_2} 
		\sum_{k=0}^\infty \widetilde{\mathcal{E}}_{k-\mu_1}(u \mu_1)
		\left (\sum_\lambda ... \right )
		\sum_{l=0}^\infty \widetilde{\mathcal{E}}_{l-\mu_2}(u \mu_2)
		\left (\sum_\nu ... \right )
		\Bigg \rangle,
	\end{align}
	where
	\begin{align}\notag
		&  \left (\sum_\lambda ... \right ) := \left ( \sum_{\lambda \vdash k} \prod_{i=1}^{l(\lambda)} \frac{1}{(\lambda_i - \lambda_{i+1})!}
		\left (\frac{A^i-A^{-i}}{i} \mu_1  b \frac{\mathcal{S}(i u \mu_1)}{\mathcal{S}(i u  b^{-1})}
		\right )^{\lambda_i - \lambda_{i+1}} \right ),
		\\ \notag
		& \left (\sum_\nu ... \right ) := \left ( \sum_{\nu \vdash l} \prod_{j=1}^{l(\nu)} \frac{1}{(\nu_i - \nu_{i+1})!}
		\left (\frac{A^j-A^{-j}}{j} \mu_2  b \frac{\mathcal{S}(j u \mu_2)}{\mathcal{S}(j u  b^{-1})}
		\right )^{\nu_j - \nu_{j+1}} \right ).
	\end{align}
	Note that we have $\widetilde{\mathcal{E}}$-operators inside the average instead of $\mathcal{E}$-operators, because the $\zeta(z)^{-1}$ summands in $\mathcal{E}_0(z)$ are canceled via the inclusion-exclusion formula that transforms the disconnected correlators into the connected ones.
	
	Thus the correlator is represented as a double sum over $k,l\geq 0$. The vacuum expectation in non equal to zero only for the summands where $k - \mu_1 > 0$ and $l - \mu_2 < 0$. Using this and the commutation formula for the $\widetilde{\mathcal{E}}$-operators~\eqref{eq:commEoper}, we see that the right hand side of Equation~\eqref{eq:2-point-in-mu} is equal to
	\begin{align} \label{eq:2-point-in-mu-commuted}
		\Bigg \langle \sum_{\substack{k,l = 0 \\ k > \mu_1 \\ l < \mu_2}}^\infty
		\frac{Q^2 A^{ b(\mu_1+\mu_2)}}{\mu_1\mu_2} \zeta(u(k \mu_2 - l \mu_1)) \mathcal{E}_{k + l - \mu_1 - \mu_2} (u (\mu_1+\mu_2))
		\left (\sum_\lambda ... \right ) \left (\sum_\nu ... \right )
		\Bigg \rangle
	\end{align}
	
	The genus 0 contribution is given by the coefficient in front of $u^0$. The average of $\mathcal{E}$-operator
	is non-zero only when $k + l - \mu_1 - \mu_2 = 0$. Thus the coefficient of $u^0$ in formula~\eqref{eq:2-point-in-mu-commuted} is equal to
	\begin{align} \label{eq:2-point-in-mu-genus-0}
		\sum_{\substack{k,l = 0 \\ k > \mu_1 \\ l < \mu_2 \\ k + l = \mu_1 + \mu_2}}^\infty
		\frac{Q^2 A^{ b(\mu_1+\mu_2)}}{\mu_1 \mu_2} \frac{(k \mu_2 - l \mu_1)}{(\mu_1 + \mu_2)}
		&
		\left (\sum_{\lambda \vdash k} \prod_{i=1}^{l(\lambda)} \frac{1}{(\lambda_i - \lambda_{i+1})!}
		\left ( \frac{A^i-A^{-i}}{i} \mu_1  b
		\right )^{\lambda_i - \lambda_{i+1}} \right )
		\\ \notag &
		\left (
		\sum_{\nu \vdash l} \prod_{j=1}^{l(\nu)} \frac{1}{(\nu_i - \nu_{i+1})!}
		\left (\frac{A^j-A^{-j}}{j} \mu_2  b
		\right )^{\nu_j - \nu_{j+1}} \right )
	\end{align}
	Taking into account all the restrictions on $k$ and $l$ as well as
	explicit expression for sums over partitions in terms of Jacobi polynomials
	(given in Corollary~\ref{prop:Jacobiseries}), we can rewrite formula~\eqref{eq:2-point-in-mu-genus-0} as
	\begin{align} \label{eq:2-point-in-mu-jacobi}
		\sum_{l = 0}^{\mu_2-1} & \frac{Q^2 (\mu_2 - l)}{(\mu_1+\mu_2-l)l} (-1)^{\mu_1+\mu_2}
		A^{( b-1)\mu_1+( b-1)\mu_2}
		(1 - A^2)^2  b^2 
		\\ \notag
		& \times \jp{\mu_1+\mu_2-l-1}{\mu_1  b - \mu_1 - \mu_2 + l}{ 1}(1 - 2 A^2)
		\jp{l - 1}{\mu_2  b - l}{ 1}(1 - 2 A^2)
	\end{align}
	
	We can further simplify this expression. Note that we have the following decomposition into
	simple fractions with respect to $l$
	\begin{align}
		\frac{(\mu_2 - l)}{(\mu_1+\mu_2-l)l} = \frac{1}{\mu_1 + \mu_2} \left ( \frac{\mu_2}{l} - \frac{\mu_1}{\mu_1+\mu_2-l} \right )
	\end{align}
	Using three-term relation~\eqref{eq:JacobiThreeTerm} in the form
	\begin{align}
		(1 - a)\dfrac{\rho b}{k} J_{k-1}^\rho = J_k^\rho+\left(a+1\right)J_{k-1}^\rho+aJ_{k-2}^\rho
	\end{align}
	where for brevity we use the notation $J_k^\rho:=J_k(\rho)=\mathcal{P}_{k}^{(\rho b-k-1,1)}(1-2a)$, $a=A^2 $,
	we can rewrite formula~\eqref{eq:2-point-in-mu-jacobi} as
	\begin{align} \label{eq:2-point-in-mu-jacobi-simplified}
		& 
		\frac{Q^2 (-1)^{\mu_1+\mu_2} A^{( b-1)\mu_1+( b-1)\mu_2}
			(1 - A^2)  b}{(\mu_1 + \mu_2)}
		\\ \notag 
		& \times 
		\sum_{l = 0}^{\mu_2} \Bigg ( {
			J_{\mu_1+\mu_2-l-1}^{\mu_1} J_l^{\mu_2}
			+ (1+a) J_{\mu_1+\mu_2-l-1}^{\mu_1} J_{l-1}^{\mu_2}
			+ a J_{\mu_1+\mu_2-l-1}^{\mu_1} J_{l-2}^{\mu_2}}
		\\ \notag
		& 
		-{J_{\mu_1+\mu_2-l}^{\mu_1} J_{l-1}^{\mu_2}
			-(1+a) J_{\mu_1+\mu_2-l-1}^{\mu_1} J_{l-1}^{\mu_2}
			- a J_{\mu_1+\mu_2-l-2}^{\mu_1} J_{l-1}^{\mu_2}}
		\Bigg )
	\end{align}
	Most of the summands in the last expression cancel and we finally obtain
	(applying the three-term relation in its original form \eqref{eq:JacobiThreeTerm}
	so that only the
	Jacobi polynomials appearing in the expansions of the $\xi$-functions remain) the following expression for the connected $2$-point correlator in genus $0$:
	\begin{align} \label{eq:2-point-in-mu-jacobi-totally-simplified}
		&
		\frac{Q^2 (-1)^{\mu_1+\mu_2} A^{( b-1)\mu_1+( b-1)\mu_2}
			(-1) (1 - A^2)  b}{(\mu_1 + \mu_2)}
		\\ \notag & \times \Bigg (
		((a + 1) + (a - 1)  b) J^{\mu_1}_{\mu_1 - 1} J^{\mu_2}_{\mu_2 - 1}
		+ a \left ( J^{\mu_1}_{\mu_1 - 1} J^{\mu_2}_{\mu_2 - 2}
		+ J^{\mu_1}_{\mu_1 - 2} J^{\mu_2}_{\mu_2 - 1} \right)
		\Bigg )
		\\ \notag
		& =  \frac{Q^2 (a - 1)  b}{(\mu_1 + \mu_2)} \Bigg (
		((a + 1) + (a - 1)  b) \xi^1_{\mu_1} \xi^1_{\mu_2}
		+ a \left ( \xi^1_{\mu_1} \xi^2_{\mu_2} + \xi^2_{\mu_1} \xi^1_{\mu_2} \right)
		\Bigg )
	\end{align}
	
\end{proof}

Denote by $E$ the Euler operator
\begin{align}
	E = \left(\Lambda_1 \frac{\p}{\p \Lambda_1} + \Lambda_2 \frac{\p}{\p \Lambda_2}\right)
\end{align}

\begin{lemma} 
	\label{lem:2-point-log}
	We have
	\begin{align}
		& E \log\br{U_1 - U_2} =  b (1 - a) \times
		\\ \notag    
		& \sum_{\mu_1, \mu_2 = 1}^\infty
		\br{
			\br{(a+1)+(a-1) b} \xi^1_{\mu_1} \xi^1_{\mu_2}
			+ a \br{\xi^1_{\mu_1} \xi^2_{\mu_2} + \xi^2_{\mu_1} \xi^1_{\mu_2}}}
		\Lambda_1^{-\mu_1} \Lambda_2^{-\mu_2}
	\end{align}
\end{lemma}

\begin{proof} Recall the expression for $\Lambda(U)$ given by Equation~\eqref{eq:bem-x-and-lambda}.
	The coefficient in front of $\Lambda_1^{-\mu_1} \Lambda_2^{-\mu_2}$ in the expansion of $ E \log\br{U_1 - U_2} $
	is equal to the following residue:
	\begin{align} \label{eq:euler-log-u}
		& [\Lambda_1^{-\mu_1} \Lambda_2^{-\mu_2}] E \log(U_1 - U_2) \\ \notag
		& =
		\frac{1}{(2 \pi \sqrt{-1})^2} \oint_{\Lambda_1} \oint_{\Lambda_2} \Lambda_1^{\mu_1-1} \Lambda_2^{\mu_2-1}
		\left(\Lambda_1 \frac{\p}{\p \Lambda_1} + \Lambda_2 \frac{\p}{\p \Lambda_2}\right) \log(U_1 - U_2) d\Lambda_1 d\Lambda_2
		\\ \notag
		& = \frac{1}{(2 \pi \sqrt{-1})^2} \oint_{\Lambda_1} \oint_{\Lambda_2} \Lambda_1^{\mu_1-1} \Lambda_2^{\mu_2-1}
		\br{\Lambda_1 \frac{d \Lambda_2}{d U_2} - \Lambda_2 \frac{d \Lambda_1}{d U_1}} \frac{1}{(U_1 - U_2)} d U_1 d U_2
		\\ \notag
		& = I_{\mu_1} (1,0) I_{\mu_2} (1,0)
		+  b \br{A^{2( b+1)} - A^{2( b-1)}} I_{\mu_1}(0,1) I_{\mu_2}(0,1)
		\\ \notag
		& \quad
		+  b (A^{ b-1} - A^{ b+1}) A^{2  b} \br{I_{\mu_1}(-1,1) I_{\mu_2}(0,1) + I_{\mu_1}(0,1) I_{\mu_2}(-1,1)},
	\end{align}
	where 
	\begin{align}
		I_\mu(x, y) : = \frac{1}{2 \pi \sqrt{-1}} \oint_{U, \infty} U^{\mu-x} \frac{(1 - A^{ b+1} U)^{ b \mu - y}}{(1 - A^{ b-1} U)^{ b \mu + y}} d U
	\end{align}
	These integrals are computed below, in Lemma \ref{lem:I-integral}. In particular, we have:
	\begin{align}
		I_\mu(1,0) & =    b\br{A^{ b-1}-A^{ b+1}} I_\mu(0,1) =  b (1 - A^2) (-1) \xi^1_\mu
		\\ \notag
		I_\mu(0,1) & =   (-1) A^{- b+1} \xi^1_\mu
		\\ \notag
		I_\mu(-1,1) & =  A^{-2  b} (-1) \bs{\br{(a+1) + (a-1) b}\xi^1_\mu + a \xi^2_\mu}
	\end{align}
	where we used integration by parts to evaluate the first integral. Substituting these expressions we obtain that the coefficient of $\Lambda_1^{-\mu_1} \Lambda_2^{-\mu_2}$ in $E \log(U_1 - U_2)$ is equal to
	\begin{align}
		 b (1 - a)
		\br{
			\br{(a+1)+(a-1) b} \xi^1_{\mu_1} \xi^1_{\mu_2}
			+ a \br{\xi^1_{\mu_1} \xi^2_{\mu_2} + \xi^2_{\mu_1} \xi^1_{\mu_2}}}
	\end{align}
	
\end{proof}

\begin{lemma} \label{lem:I-integral} We have:
	\begin{align} 
		& I_\mu(x, y) := \frac{1}{2 \pi I} \oint_{U, \infty} U^{\mu-x} \frac{(1 - A_+ U)^{b \mu - y}}{(1 - A_- U)^{b \mu + y}} d U 
		\\ \notag
		& = \frac{A^{2 b(\mu - y)} (-1)^{\mu-x-2 y + 1}}{A_+^{\mu-x-2 y + 1}}
		\jp{\mu-x-2 y + 1}{\mu(b-1) + x + y - 1}{2 y - 1} (1 - 2 A^2).
	\end{align}
\end{lemma}

\begin{proof}
	Explicit calculation using expansion obtained in Lemma \ref{lem:shifted-genfunc-expansion}.
\end{proof}

\section{Quantum curve} \label{sec:qsc}

In this Section we derive a quantum spectral curve for a natural wave function obtained from the extended Ooguri-Vafa partition function given by the vacuum expectation formula~\eqref{eq:extended-partition-function}. In this derivation we follow closely the ideas used in~\cite{Zhou,MulShaSpi}.

A natural wave function $\Psi(\Lambda)$ is obtained by the substitution $\pb_i:=Q^{-1}\Lambda^i$ in $Z^{\ext}$ restricted to the topological locus with the simultaneous change $\hbar\to -\hbar$ (cf.~the general formulas and computation for low genera above --- we need a sign adjustment of $(-1)^n$ for the $n$-point functions, which can be achieved by replacing $\hbar^{2g-2+n}$ by $(-\hbar)^{2g-2+n}$). We have:
 \begin{align*} 
\Psi(\Lambda) &=
\left\langle
\exp \lb \sum_{j=1}^\infty \frac{\alpha_{-j}}{j} \cdot \frac{A^j-A^{-j}}{\zeta(-j\hbar)} \rb
\exp \lb -\hbar b \Fc_2 \rb
\exp \lb \sum_{i=1}^\infty \frac{\alpha_{i} \Lambda^{i}A^{ib}}{i} \rb
\right\rangle
\\ 
& =\sum_{\ell=0}^\infty \Lambda^{\ell} A^{\ell b} \exp\left(-\hbar b \frac {\ell^2-\ell}{2} \right) 
\prod_{i=1}^\ell \frac{Ae^{-\hbar (i-1)/2} - A^{-1}e^{\hbar (i-1)/2} }{e^{-\hbar i/2} - e^{\hbar i/2}}
\end{align*}
Here we used that $s_R|_{p_i=r^i}$ is non-trivial only for $R=(\ell,0,0,\dots)$, in which case it is equal to $r^\ell$, and the fact that 
\[
s_\ell^*:=s_\ell|_{p_i=p_i^*} = \prod_{i=1}^\ell \frac{Ae^{\hbar (i-1)/2} - A^{-1}e^{-\hbar (i-1)/2} }{e^{\hbar i/2} - e^{-\hbar i/2}}.
\]

\begin{theorem} We have:
\begin{equation}\label{eq:qsc}
\left[\left(
e^{-\frac\hbar2 \Lambda\frac d {d\Lambda}} - e^{\frac\hbar2 \Lambda\frac d {d\Lambda}}
\right)
-
\Lambda A^b e^{-\hbar b \Lambda\frac d {d\Lambda}} \left(
Ae^{-\frac\hbar2 \Lambda\frac d {d\Lambda}} - A^{-1}e^{\frac\hbar2 \Lambda\frac d {d\Lambda}} 
\right)\right]
\Psi = 0.
\end{equation}
\end{theorem}

\begin{proof}
We denote the $\ell$-th term in the representation of $\Psi$ above by $\psi_\ell$, that is, $\Psi=\sum_{\ell=0}^\infty \psi_\ell$. Note that 
\[
\frac{\psi_{\ell+1}}{\psi_\ell} = \Lambda A^b e^{\hbar b \ell} \frac{Ae^{-\hbar \ell/2} - A^{-1}e^{\hbar \ell/2} }{e^{-\hbar (\ell+1)/2} - e^{\hbar (\ell +1)/2}},
\]
or, in other words,
\[
\left(
e^{-\frac\hbar2 \Lambda\frac d {d\Lambda}} - e^{\frac\hbar2 \Lambda\frac d {d\Lambda}}
\right)
\psi_{\ell+1} = \Lambda A^b e^{-\hbar b \Lambda\frac d {d\Lambda}} \left(
Ae^{-\frac\hbar2 \Lambda\frac d {d\Lambda}} - A^{-1}e^{\frac\hbar2 \Lambda\frac d {d\Lambda}} 
\right)\psi_{\ell},
\]
which implies the statement of the theorem.
\end{proof}

The differential-difference operator annihilating the wave function $\Psi$~\eqref{eq:qsc} is the quantum spectral curve. Under the dequantization $\exp (\hbar \Lambda\frac d{d\Lambda})\to V$ we obtain the equation 
\[
\Lambda = \frac{1-V}{A^{b+1}V^{-b} (1-A^{-2}V)},
\]
which is an equivalent way to present the BEM spectral curve~\eqref{eq:bem-curve} (in this form it is given in~\cite[Equation (3.26)]{BEM11}).

\end{document}